\documentclass[accepted=2022-05-30,a4paper,pra,floatfix,superscriptaddress,showpacsonecolumn,nofootinbib,longbibliography]{quantumarticle}
\usepackage{amsmath, amsthm, amssymb,amsfonts,mathbbol,amstext}
\usepackage{graphicx}
\usepackage{dcolumn}
\usepackage{bm}
\usepackage{bbm}
\usepackage{hyperref}
\usepackage{mathtools}
\usepackage{color}
\usepackage{multirow}
\usepackage{diagbox}
\usepackage{float}
\usepackage{xcolor}
\usepackage{listings}

\usepackage{changes}

\def\1{\mathbf{1}}
\def\0{\mathbf{0}}

\newcommand*{\coloneqq}{\mathrel{\vcenter{\baselineskip0.5ex \lineskiplimit0pt \hbox{\scriptsize.}\hbox{\scriptsize.}}} =}

\DeclareMathOperator{\tr}{tr}
\DeclareMathOperator{\Tr}{Tr}

\newcommand{\ket}[1]{| #1 \rangle}
\newcommand{\bra}[1]{\langle #1 |}

\newtheorem{prop}{Proposition}

\newtheorem{proposition}[prop]{Proposition}

\newcommand{\processnext}[1]{%
	\ifx\listfinish#1\empty\else\listact{#1}\expandafter\processnext\fi}

\newcommand{\bdm}[1]{d^{\,2m}\!\bm{#1}\,}

\newcommand{\SQPD}{\hbox{($\bm{S}$)-PQD}}

\newcommand{\ignore}[1]{}
\newcommand{\nobibentry}[1]{{\let\nocite\ignore\bibentry{#1}}}

\newcommand{\average}[1]{\left<#1\right>}
\let\oldsqrt\sqrt
\def\sqrt{\mathpalette\DHLhksqrt}
\def\DHLhksqrt#1#2{\setbox0=\hbox{$#1\oldsqrt{#2\,}$}\dimen0=\ht0
	\advance\dimen0-0.2\ht0
	\setbox2=\hbox{\vrule height\ht0 depth -\dimen0}%
	{\box0\lower0.4pt\box2}}

\begin{document}
\title{Thermometry of Gaussian quantum systems using Gaussian measurements}

\author{Marina F.B. Cenni}
\address{ICFO-Institut de Ciencies Fotoniques, The Barcelona Institute of Science and Technology, 08860 Castelldefels (Barcelona), Spain}
\author{Ludovico Lami}
\address{Institut f\"{u}r Theoretische Physik und IQST, Universit\"{a}t Ulm, Albert-Einstein-Allee 11, D-89069 Ulm, Germany}
\author{Antonio Ac\'in}
\address{ICFO-Institut de Ciencies Fotoniques, The Barcelona Institute of
Science and Technology, 08860 Castelldefels (Barcelona), Spain}
\address{ICREA-Instituci\'o Catalana de Recerca i Estudis Avan\c{c}ats, 08010, Barcelona, Spain}
\author{Mohammad Mehboudi}
\address{D\'epartement de Physique Appliqu\'ee, Universit\'e de Gen\`eve, 1205 Gen\`eve, Switzerland}
\begin{abstract}
We study the problem of estimating the temperature of Gaussian systems with feasible measurements, namely Gaussian and photo-detection-like measurements. For Gaussian measurements, we develop a general method to identify the optimal measurement numerically, and derive the analytical solutions in some relevant cases.
%Generally, the answer can be found numerically, however in special cases it can be done analytically as well.
For a class of single-mode states that includes thermal ones, the optimal Gaussian measurement is either Heterodyne or Homodyne, depending on the temperature regime. This is in contrast to the general setting, in which a projective measurement in the eigenbasis of the Hamiltonian is optimal regardless of temperature. In the general multi-mode case, and unlike the general unrestricted scenario where joint measurements are not helpful for thermometry (nor for any parameter estimation task), it is open whether joint Gaussian measurements provide an advantage over local ones. We conjecture that they are not useful for thermal systems, supported by partial analytical and numerical evidence. We further show that Gaussian measurements become optimal in the limit of large temperatures, while {on/off} photo-detection-like measurements do it for when the temperature tends to zero. Our results therefore pave the way for effective thermometry of Gaussian quantum systems using \textit{experimentally realizable measurements}.
\end{abstract}
\maketitle
%\newpage
%\onecolumngrid
\tableofcontents
\section{Introduction and motivation}\label{intro}
%
%\moha{look at this too: \hyperlink{https://iopscience.iop.org/article/10.1088/2058-9565/abd83d/pdf}{https://iopscience.iop.org/article/10.1088/2058-9565/abd83d/pdf}}\\
The high precision required for the operation of quantum devices demands for the estimation of their parameters with minimum possible error.
In particular, estimating temperature is a problem that has attained substantial popularity in recent years~\cite{Mehboudi_2019,de2018quantum}, since it is a parameter that highly impacts the technological usage of quantum systems. 
Thermometry process, like any other parameter estimation problem, can be seen as different stages of (i) preparation of a probe or thermometer, (ii) interaction of the probe with the sample, and (iii) measurement of the probe and data analysis~\cite{giovannetti2011advances,PhysRevLett.96.010401,T_th_2014}, see also Fig.~\ref{fig:my_label}. 
When there are no restrictions on any of these stages, one finds the \textit{ultimate} limits on thermometry precision~\cite{Mehboudi_2019,de2018quantum}. 
These limits have deeply increased our understanding of fundamental restrictions/possibilities in any thermometry procedure~\cite{PhysRevLett.114.220405}. Nonetheless, one also needs to address frameworks that are experimentally feasible, hence revisit the fundamental limits subject to experimental restrictions on the different stages above. In this work, we are in particular concerned with restrictions on the performed measurements, stage (iii) above.
%\deleted{, i.e., we are restricted to a specific set of measurements}. 
To our knowledge, there have been so far few works addressing measurement restrictions, such as limited resolution~\cite{PhysRevResearch.2.033394} or coarse grained measurements~\cite{PRXQuantum.2.020322}, in thermometry applications. 

Our main goal is to study the use of experimentally feasible measurements in thermometry processes using continuous variable (CV) Gaussian systems.
%\deleted{ subject to experimental limitations on the measurement stage}. 
CV Gaussian models are very successful at describing various physical systems of relevance, such as Bosonic gases, electromagnetic fields, {mechanical oscillators} and Josephson junctions, to name a few. This makes them versatile candidates for several quantum information processing as well as quantum simulation tasks~\cite{RevModPhys.84.621, doi:10.1080/00018730701223200, lewenstein2012ultracold, serafini2017quantum, e17096072, PhysRevLett.94.020505}. When there is no limitation on the performed measurements, thermometry (and more generally metrology) of Gaussian CV systems is characterized readily, in particular the optimal measurement and the ultimate precision can be found analytically~\cite{monras2013phase,PhysRevA.89.032128,PhysRevA.98.012114}. Nonetheless, a rigorous analysis of quantum metrology with realistic measurements, namely Gaussian as well as {two-outcome} photo-detection measurements, is missing. Our work addresses this question in the context of temperature estimation. 
%\replaced{}{The aim of this work is to partly answer this question, specifically for temperature estimation}.

Several previous theoretical works have already focused on thermometry via Gaussian probes, e.g., exploiting Bosonic impurities embedded in cold gases, or light modes that interact with an external field~\cite{PhysRevLett.122.030403,planella2020bath,PhysRevA.96.062103,PhysRevB.98.045101,mukherjee2019enhanced,Mart_n_Mart_nez_2013}. Unfortunately, in these systems the optimal measurement for thermometry can be experimentally very demanding and effectively out of reach. For example, consider a scenario where a light mode is used as thermometer. After interacting with a sample with an unknown temperature $T$, the light mode can relax to a thermal state with the same temperature. By performing measurements on the light mode one aims to estimate the temperature. The maximally informative measurement in this case is an energy measurement. This corresponds to ideal \textit{photon-counting}, which is extremely difficult. It is thus advisable to investigate experimentally realisable measurements, such as realistic {on/off} photo-detection or Gaussian measurements and benchmark their performance by comparing them with the known optimal measurement. In this work, we follow this approach and first consider single-mode systems. We analytically prove that the optimal Gaussian measurement corresponds to a homodyne or heterodyne measurement, depending on the temperature. We investigate the conditions under which the ultimate bound on thermometry precision is attained.
The appendix is there for self-containedness....
We find that at high temperatures Gaussian measurements are almost optimal, while they perform poorly in the low-temperature regime, where, in turn, {on/off} photo-detection is close to optimal.
%In particular, we see that at high $T$ Gaussian measurements are optimal. While at low $T$ they have a disappointing performance, ideal photo-detection is optimal at this regime. 
We then move to the multi-mode scenario. While in the case where measurements are not restricted the additivity of the Fisher information implies that the optimal measurement consist of the independent application of the optimal measurement for single modes, whether the same holds true when restricting the measurements to be Gaussian is open.
%\replaced{}{Furthermore, we study the best Gaussian measurements in the multimode scenario %\tcr{[LL: do we mean multi-copy scenario? A single system can have multiple modes...]}. 
%We show that unlike %the ultimate bound (i.e., without restricting to Gaussian measurements)
%in the unrestricted case, this problem is non-trivial. When optimising over all measurements, the optimal measurement turns out to be local---due to additivity of the quantum Fisher information. However, when restricting to Gaussian measurements, it is not clear whether joint measurements on all modes are beneficial or not.} 
We conjecture---supported by numerical evidence---that this is indeed the case, and analytically prove it for special cases. 
%Finally, we address the equivalence between the error-propagation formula and the mean-square error for homodyne like measurements. \added{DO YOU WANT TO INCLUDE THIS, IMO, RATHER SPECIFIC RESULT HERE?}

\begin{figure}
    \centering
    \includegraphics[width=0.99\linewidth]{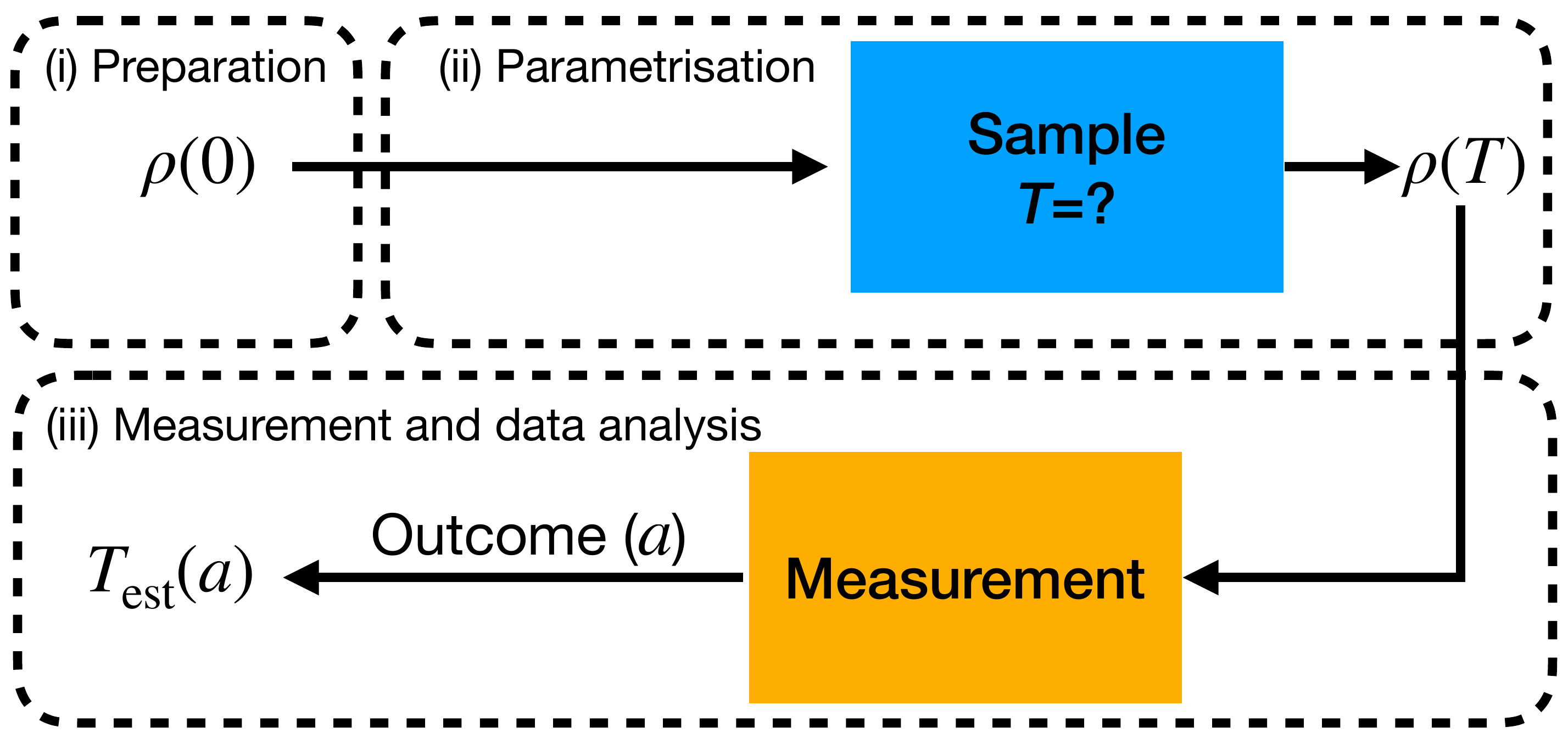}
    \caption{Thermometry protocols can be broken into three parts of (i) preparation of a probe, (ii) interaction with the sample and parametrisation, and (iii) measurement and data analysis. The focus of this work is on the last part of the protocol (iii). In particular, we explore the precision bounds restricted to experimentally feasible measurements, namely Gaussian or {on/off} photo-detection measurements.}
    \label{fig:my_label}
\end{figure}

\section{Preliminaries: The Phase-space formalism}\label{2}

We start by setting up the notation and introducing the necessary tools from the phase-space formalism, such as Gaussian systems and Gaussian measurements.
%\deleted{, before we proceed further}. 
For additional details, we refer the interested reader to the more complete references~\cite{HOLEVO,BARNETT-RADMORE, serafini2017quantum}. The basics of parameter estimation in Gaussian systems is discussed in the next section, where we also present some of our results.

\subsection{Gaussian States and Measurements}\label{sec:Gaussian_states_measure}

Every trace class\footnote{A linear bounded operator $T:\mathcal{H}\to \mathcal{H}$ acting on a separable Hilbert space $\mathcal{H}$ is said to be of trace class if $\Tr|T|<\infty$, where $|T|\coloneqq \sqrt{T^\dag T}$, and for a positive semidefinite $A\geqslant 0$ the formal trace $\Tr A\coloneqq \sum_i \bra{i}A\ket{i}$ can be shown not to depend on the chosen orthonormal basis $\{\ket{i}\}_i$ of $\mathcal{H}$.} operator defined in a Hilbert space has an equivalent representation in terms of $\bm{S}$-ordered Phase-space Quasi-probability Distributions [\SQPD s], which are distributions defined over a real symplectic space that is referred to as phase-space. In order to characterize this representation, we first define the frame operators. These operators are denoted by $\Delta(\bm z)$ and are often used to define the $\bm{S}$-ordered quasi-probability distributions $W( {O},\bm z)\coloneqq \tr [{O}\Delta(\bm z)]$ corresponding to the operator $ {O}$, that can either be an observable or a density operator in the Hilbert space\footnote{Throughout this paper we use a notation in which phase space related quantities are in bold font, e.g., ${\bm \sigma}$ and ${\bm d}$ represent the covariance matrix and displacement vector. On the contrary, we write the operators in the Hilbert space of the system in normal font, e.g., ${\rho}$ represents the density matrix. For the positive operator-valued measure (POVM) elements like ${M_{{\bm a}|{s}}}$ the operator itself is in the Hilbert space and represented with normal font, but the outcome is often a phase-space vector like ${\bm a}$ thus we use the bold font.}.

First we note that, for a given quantum state $\rho$ and any set of POVM operators $\{{M}_{{{\bm a}|s}}\}$---with ${s}$ labeling the specific choice of POVM and ${\bm a}$ the different outcomes for that given choice---the outcome probabilities for the measurements are given by the Born rule,
\begin{equation}
	P(\bm{a}|\rho;s) = \tr [\rho {M}_{\bm{a}|{s}}].\label{eq:Born_rule}
\end{equation}
To show how any $m$-mode operator like $\rho$ can be represented by \SQPD s in phase-space, we start by expanding $\rho$ in terms of the canonical displacement operators, %(CITE?)
\begin{align}\label{eq:spqd_1}
		\rho & = \frac{1}{(2\pi)^{m}}\int d^{2m}\bm{y} \tr [\rho D(\bm{y})] D(-\bm{y})
\end{align}
where $D(\bm y) \coloneqq \exp(-i\bm{y}^{\intercal} \bm{\Omega} R)$ is the displacement (or Weyl) operator, $R^{{\intercal}}=(q_1,p_1,\dots, q_m,p_m)$ is the vector of canonical quadrature operators, $[q_j,p_k]=i\delta_{j,k}$, and $\bm{y}, \bm{z} \in \mathbb{R}^{2m}$, and $\bm\Omega=\bigoplus_{j=1}^{m}\left(\begin{smallmatrix}0 & 1\\ -1 & 0\end{smallmatrix}\right)$ is the symplectic form enforcing the uncertainty relations in phase-space.

The $\Delta(\bm z)$ operators are defined as the Fourier transforms of the displacement operator
\begin{equation}
	\Delta(\bm z) \coloneqq \int\frac{\bdm{y}}{(2\pi)^{2m}}\,
	D(\bm y)\,
	e^{-i\bm{z}^{{\intercal}}\bm{\Omega}\bm{y}} ,
\end{equation}
where the integral is intended to converge in the weak operator topology. The Wigner function for the density operator $\rho$ is then given by
\begin{align}
	W(\bm{z}|\rho) & = \tr \big[\rho\Delta(\bm z)\big] = \int\frac{\bdm{y}}{(2\pi)^{2m}}\,
	\tr [\rho D(\bm y)]\,
	e^{-i\bm{z}^{{\intercal}}\bm{\Omega}\bm{y}}.
	\label{WS}
\end{align}

Since $\tr [D(\bm y)]=(2\pi)^m\delta^{2m}(\bm{y})$, the operator $\Delta(\bm z)$ satisfies the following relations:
\begin{equation}
	\tr [\Delta(\bm z)]=\int\frac{\bdm{y}}{(2\pi)^{2m}}\,
	\tr [D(\bm y)]\,
	e^{-i\bm{z}^{{\intercal}}\bm{\Omega}\bm{y}}=\frac{1}{(2\pi)^m},
	\label{PQD-identity}
\end{equation}
and 
\begin{equation}
	\int d^{2m}\bm{z}\, \Delta(\bm z)\hspace{-1mm}=\hspace{-2mm} \int\frac{\bdm{y}}{(2\pi)^{2m}}\,
	D(\bm y)\,
	\hspace{-2mm}
	\int d^{2m} \bm{z} e^{-i\bm{z}^{{\intercal}}\bm{\Omega}\bm{y}}=I,
\end{equation}
where we have used that
\begin{align}
	\int d^{2m} \bm{z} e^{-i\bm{z}^{\intercal}\bm{\Omega}\bm{y}} = (2\pi)^{2m} \delta^{2m}(\bm{y}).\label{Four-delta-func}
\end{align}

By taking the inverse Fourier transform of Eq.~(\ref{WS}) and using Eq.~(\ref{Four-delta-func}), one obtains 
\begin{align}
	\int d^{2m}\bm{z}\, W(\bm{z}|\rho) e^{i\bm{z}^{{\intercal}}\bm{\Omega}\bm{y}}=
	\tr [\rho D(\bm{y})]. \label{eq:spqd_2}
\end{align}
Substituting Eq. \eqref{eq:spqd_2} into Eq. \eqref{eq:spqd_1} gives
\begin{align}\label{eq:DM}
	\begin{split}
		\rho & = \frac{1}{(2\pi)^{m}}\int d^{2m} \bm{z}\, W(\bm{z}|\rho) \int d^{2m}\bm{y}  D(-\bm{y})\, e^{i\bm{y}^{{\intercal}}\Omega \bm{y}} \\
		&= (2\pi)^{2m}\int d^{2m} \bm{z}\, W(\bm{z}|\rho) \Delta(\bm{z}).
	\end{split}
\end{align}

Finally, by using these results in Eq.~\eqref{eq:Born_rule}, we obtain
\begin{align}\label{eq:Born_phase_space}
	%	p(a|x)&=
	P(\bm{a}|\rho;s)&=\tr [\rho {M}_{\bm{a}|{s}}]\nonumber\\
	&=(2\pi)^M \int \bdm{z} W(\bm z|\rho)\, W(\bm{a}|{s},\bm z),
\end{align}
where $W(\bm{a}|{s},\bm z)=\tr \big[{M}_{\bm{a}|{s}} \Delta(\bm{z})\big]$ is the Wigner function associated to the POVM measurement operators.\footnote{The only measurement operators we will encounter will be of trace class.} From this relation, we see that the information contained in the probabilities given by Eq.~\eqref{eq:Born_rule} is equivalently contained in the integral overlap of the corresponding Wigner functions. 
\subsubsection{The overlap of two systems}
For any two quantum states $\rho_1$ and $\rho_2$  we can use \eqref{eq:spqd_1} to compute their overlap
\begin{align}
\tr [\rho_1 \rho_2] 
% & = (2\pi)^{m} \int d^{2m}\bm{z} W(\rho_1|\bm{z}) W(\bm{z}|\rho_{2})\nonumber\\
% & = \frac{1}{(2\pi)^{3m}}\iint\int d^{2m} \bm{z} d^{2m}\bm{y} d^{2m}\bm{x} \tr [\rho_1 D(\bm{y})]\tr [\rho_2 D(\bm{x})] e^{-i\bm{z}\bm{\Omega}\bm{y}^\intercal}e^{-i\bm{z}\bm{\Omega}\bm{x}^\intercal}\nonumber\\
% & = \frac{1}{(2\pi)^{m}}\iint d^{2m}\bm{y} d^{2m}\bm{x} \tr [\rho_1 D(\bm{y})]\tr [\rho_2 D(\bm{x})]\delta^{2m}(\bm{y}+\bm{x})\nonumber\\
& = \frac{1}{(2\pi)^{m}}\int d^{2m}\bm{y} \tr [\rho_1 D(\bm{y})]\tr [\rho_2 D(-\bm{y})],
\end{align}
Now, the term $\tr [D(\bm{y})\rho_i]$ is nothing but the characteristic function of the $i$th system. That is, by definition,
\begin{align}
\chi(\bm{y}|\rho_i) \coloneqq \tr [D(\bm{y})\rho_i].
\end{align}
For Gaussian systems $\rho_i$ with covariance matrices ${\bm \sigma}_i$ and displacement operators ${\bm d}_i$, the characteristic function is given by the Gaussian distribution
\begin{align}
\chi(\bm{y}|\rho_i) = e^{\frac{1}{2}\bm{y}^{{\intercal}}\bm{\Omega} \bm{\sigma}_i \bm{\Omega} \bm{y} + i \bm{d}_i^{{\intercal}}\bm{\Omega} \bm{y}},
\end{align}
with the displacement vector and the covariance matrix (CM) defined as $\bm{d}_i|_k = \average{R_k}_{\rho_i}$, and $\bm{\sigma}_i|_{kl} = \average{\{R_k,R_l\}_+}_{\rho_i}/2-\average{R_k}_{\rho_i}\average{R_l}_{\rho_i}$, where $\average{A}_{\rho_i} \coloneqq \tr [A\rho_i]$ and $\{ , \}_+$ represents the anti-commutator. 
%[Notice that some authors define the CM as twice this definition, and therefore in the definition of the characteristic function the $-1/2$ is replaced by $-1/4$. Note also that in the paper of R.Nichols they make a mistake in their eq.4, with their definitions this should be -1/4.\textbf{(Nina: do we include this note?)}].
Having this form for the Gaussian characteristic function, we can write
\begin{align}
\tr [\rho_1 \rho_2] & =\frac{1}{(2\pi)^{m}}\int d^{2m}\bm{y} \tr [\rho_1 D(\bm{y})]\tr [\rho_2 D(-\bm{y})]\nonumber\\
& = \frac{1}{(2\pi)^{m}}\int d^{2m}\bm{y} e^{\frac{1}{2}\bm{y}^\intercal\bm{\Omega} (\bm{\sigma}_1+\bm{\sigma}_2) \bm{\Omega} \bm{y}}\nonumber\\
&~~\times e^{ i (\bm{d}_1-\bm{d}_2)^\intercal\bm{\Omega} \bm{y} }\nonumber\\
& = \frac{1}{(2\pi)^{m}}\int d^{2m}\bm{y} e^{\frac{1}{2}\bm{y}^\intercal{\tilde{\bm{\sigma}}} \bm{y} + i {\tilde {\bm{d}}}^\intercal \bm{y} },
\end{align}
where we define
${\tilde{\bm{\sigma}}} \coloneqq \bm{\Omega} (\bm{\sigma}_1+\bm{\sigma}_2) \bm{\Omega}$, and ${\tilde{\bm{d}}}^\intercal \coloneqq (\bm{d}_1-\bm{d}_2)^\intercal\bm{\Omega}$.
In order to solve this last integral, we need to find the orthogonal basis of ${\tilde{\bm{\sigma}}}$, such that the integration variables decouple. Performing this elementary transformation yields
%We use the final result from \href{https://en.wikipedia.org/wiki/Common_integrals_in_quantum_field_theory}{here} to obtain
\begin{align}\label{eq:Gaussian_Multid_int}
\tr [\rho_1 \rho_2] & = \frac{1}{(2\pi)^{m}}\int d^{2m}\bm{y} e^{\frac{1}{2}\bm{y}^\intercal{\tilde{\bm{\sigma}}} \bm{y} + i {\tilde {\bm{d}}}^\intercal \bm{y} } \nonumber\\
&= \frac{1}{(2\pi)^{m}}\sqrt{\frac{(2\pi)^{2m}}{{\rm det} {\tilde {\bm{\sigma}}}}} e^{\frac{1}{2}\tilde{\bm{d}}^\intercal{\tilde{\bm{\sigma}}}^{-1} \tilde{\bm{d}}} \nonumber\\
& = \frac{e^{-\frac{1}{2}{(\bm{d}_1-\bm{d}_2)^\intercal}{(\bm{\sigma}_1+\bm{\sigma}_2)}^{-1}{(\bm{d}_1-\bm{d}_2)}}}{\sqrt{{\rm det} (\bm{\sigma}_1+\bm{\sigma}_2)}} .
\end{align} 
The Gaussian measurement operators of interest here can be written as displaced Gaussian states $\rho_{s}^M$, with the displacement vector and the covariance matrix given by ${\bm d}_{s}^M$ and ${\bm \sigma}_{s}^M$, respectively.\footnote{When necessary we use the superscript $M$ to refer to measurement operators, displacement vectors, and covariance matrices, to prevent confusion with those of the system.} We can write
\begin{align}\label{eq:simple_Gaussian}
	{M}_{\bm{a}|{s}} = D(\bm{a}) \rho_{s}^M D^{\dagger}(\bm{a})/(2\pi)^m.
\end{align}
Our aim is then to find the characteristic function of this Gaussian POVM element. By using the cyclic property of the trace, and the properties of the displacement operator $D(\bm{y})D(\bm{a}) = \exp[-i\bm{y}^\intercal\bm{\Omega}\bm{a}/2]D(\bm{y}+\bm{a})$, and $D^{\dagger}(\bm{y}) = D(-\bm{y})$, we have that
\begin{align}
\chi(\bm{y}|{M}_{\bm{a}|{s}}) & = \tr  [D(\bm{a}) \rho_{s}^M D^{\dagger}(\bm{a}) D(\bm{y})]/(2\pi)^m\nonumber\\
& = \frac{e^{i\bm{a}^\intercal\bm{\Omega}\bm{y}}}{(2\pi)^m} \chi(\bm{y}|\rho_{s}^M),
\end{align}
which means that---up to a $1/(2\pi)^m$ constant---${M}_{\bm{a}|{s}}$ has a Gaussian characteristic function with the same covariance matrix as $\rho_{s}^M$, but with displacement vector increased by the amount $\bm{a}$ referring to the measurement outcome.
By putting the pieces together, we have that, for a Gaussian measurement performed on a Gaussian state $\rho$ with displacement vector $\bm{d}$ and CM $\bm{\sigma}$, the outcome probability distribution takes the (also Gaussian) form
\begin{align}\label{eq:prob_comp}
P(\bm{a}|\rho;s) &= \tr [\rho {M}_{\bm{a}|{s}}] \nonumber\\
&= \frac{e^{-\frac{1}{2}{(\bm{d}_{s}^{M}+\bm{a}-\bm{d})^\intercal}{({\bm{\sigma}_{s}^M}+\bm{\sigma})}^{-1}{(\bm{d}_{s}^{M}+\bm{a}-\bm{d})}}}{(2\pi)^{m}\sqrt{{\rm det} ({\bm{\sigma}_{s}^M}+\bm{\sigma})}},
\end{align}
where ${\bm d}_{s}^M$ and ${\bm \sigma}_{s}^M$ are respectively the displacement vector and the covariance matrix of  ${\rho}_{s}^{M}$, i.e., the Gaussian state of the measurement. The above is the most general form of the probability distribution of the outcome of a Gaussian measurement on a Gaussian state; however, in what follows we shall assume for simplicity that both $\rho$ and $\rho_{s}^M$ have vanishing displacements, so that
\begin{align}\label{gdistrib}
P(\bm{a}|\rho;s) = \frac{1}{(2\pi)^{m}\sqrt{{\rm det} ({\bm{\sigma}_{s}^M}+\bm{\sigma})}} e^{-\frac{1}{2}\bm{a}^\intercal{({\bm{\sigma}_{s}^M}+\bm{\sigma})}^{-1} \bm{a}}.
\end{align} 

The normalisation of this probability can be checked easily
\begin{align}
\int d^{2m}\bm{a} P(\bm{a}|\rho;s)  & = \frac{\int d^{2m}\bm{a} e^{-\frac{1}{2}\bm{a}^\intercal{({\bm{\sigma}_{s}^M}+\bm{\sigma})}^{-1} \bm{a}}}{(2\pi)^{m}\sqrt{{\rm det} ({\bm{\sigma}_{s}^M}+\bm{\sigma})}}\nonumber\\
& = \frac{\sqrt{\frac{(2\pi)^{2m}}{{\rm det} ({\bm{\sigma}_{s}^M}+\bm{\sigma})^{-1}}} }{(2\pi)^{m}\sqrt{{\rm det} ({\bm{\sigma}_{s}^M}+\bm{\sigma})}}= 1.
\end{align}
See the appendix for a generalisation of Gaussian measurements that incorporates noise and auxiliary modes.

\section{Thermometry of Gaussian quantum systems: Optimal vs feasible measurements}
As already advanced, we are interested in the measurement part of thermometry protocols. Let us first recall some basics of metrology in quantum systems. Given a quantum system we aim to estimate a parameter like temperature $T$ imprinted on its density matrix denoted by $\rho(T)$. To this end, one performs a POVM measurement with positive elements $\{M_{{\bm a}|{s}}\}$, where $s$ labels the specific measurement and ${\bm a}$ the different outcomes for that specific measurement. We have $\int d{\bm a} M_{{\bm a}|{s}} = I$ for all $s$, with $I$ being the identity operator. 
The outcomes of the measurement are then mapped into an estimate of the true parameter by using an estimator function $T_{\rm est}({\bm a})$. One can quantify the error on such an estimation process by means of the mean square error $\delta^{2} T_{\rm est}(\rho;s) \coloneqq \int d{\bm a} P(\bm{a}|\rho;s)(T_{\rm est}({\bm a}) - T)^2$, with $P(\bm{a}|\rho;s)=\tr [M_{{\bm a}|{s}}\rho(T)]$ being the probability of observing the outcome ${\bm a}$ conditioned on measuring $s$. Notice that the estimation error depends on the performed measurement.  
It is well known that for any unbiased estimator, that is, for any estimator satisfying $\int d{\bm a} P(\bm{a}|\rho;s) T_{\rm est}({\bm a}) = T$, the mean square error is lower bounded by the Cram\'{e}r-Rao bound
\begin{align}\label{eq:CRB}
    \delta^{2} T_{\rm est}(\rho;s)\geqslant \frac{1}{m {\cal F}^{\rm C}(\rho;s)},
\end{align}
with ${\cal F}^{\rm C}(\rho;s)$ being the (classical) Fisher information given by
\begin{align}\label{eq:CFI_general}
    {\cal F}^{\rm C}(\rho;s) & \coloneqq \average{\left[\partial_{T}\log P(\bm{a}|\rho;s) \right]^2} \nonumber\\
    &= \int d{\bm a} P(\bm{a}|\rho;s)  \left[\partial_{T}\log P(\bm{a}|\rho;s) \right]^2.
\end{align}
The Cram\'{e}r-Rao bound can be saturated in the asymptotic limit by choosing the maximum likelihood estimator~\cite{newey1994large}. %\tcr{[LL: missing reference here!]}
A good candidate to quantify the precision of measurements is therefore given by the Fisher information; the larger, the more precise the measurement. 

It is a natural and fundamental question to ask: what is the ultimate bound on precision? This is obtained by finding the measurement that minimises Eq.~\eqref{eq:CRB}. The corresponding precision is quantified by the so-called quantum Fisher information (QFI), which can be formally defined as 
\begin{align}
    {\cal F}^{\rm Q}(\rho) \coloneqq \max_{s}{\cal F}^{\rm C}(\rho;s) = \tr [\rho \Lambda^2].
\end{align}
Here, $\Lambda$ is the symmetric logarithmic derivative (SLD), a Hermitian operator defined by
\begin{align}
    2\partial_T \rho(T) = \Lambda \rho(T) + \rho(T) \Lambda.
\end{align}
The SLD also characterizes the optimal measurement, which turns out to be projective in the basis of $\Lambda$.
\subsection{Ultimate bounds on thermometry of Gaussian systems}
The problem of finding the ultimate bounds in parameter estimation of Gaussian quantum systems was first addressed in~\cite{monras2013phase,PhysRevA.88.040102}, further explored in~\cite{PhysRevA.89.032128}, and extended to multiple parameter estimation in~\cite{PhysRevA.98.012114}. It was shown that the optimal measurement can be always written as a second order operator in terms of the quadratures, that is
\begin{align}\label{eq:SLD}
    \Lambda = \sum_{kl}{\bm C}^{(2)}_{kl}R_kR_l + \sum_{k}{\bm C}^{(1)}_{k}R_k + {\bm C}^{(0)},
\end{align}
where ${\bm C}^{(2)}$ is a $2N\times 2N$ matrix of real numbers, ${\bm C}^{(1)}$ is a $2N$ dimensional real vector, and ${\bm C}^{(0)}$ is a real number. These three quantities are given by the solution of the following equations:
\begin{align}
    \partial_{T}{\bm \sigma} & = 2{\bm \sigma}{\bm C}^{(2)}{\bm \sigma} + \frac{1}{2}{\bm \Omega}{\bm C}^{(2)}{\bm \Omega},\\
    {\bm C}^{(1)} 
    &= {\bm \sigma}^{-1}\partial_{T}{\bm d} - 2{\bm C}^{(2)}{\bm d},\\
    {\bm C}^{(0)}
    &= -{\bm C}^{(1)T}{\bm d} - \frac{1}{2}\Tr [{\bm C}^{(2)}{\bm \sigma}] - {\bm d}^\intercal{\bm C}^{(2)}{\bm d}.
\end{align}
In turn, the quantum Fisher information reads
\begin{align}
    \hspace{-2mm}{\cal F}^{\rm Q}({\bm \sigma},{\bm d})  & = \partial_{T} {\bm d}^{T}{\bm \sigma}^{-1}\partial_{T}{\bm d} \nonumber\\
    &+ 2\Tr [{\bm C}^{(2)}{\bm \sigma}{\bm C}^{(2)}{\bm \sigma} + \frac{1}{4}{\bm C}^{(2)}{\bm \Omega}{\bm C}^{(2)} {\bm \Omega}].
\end{align}
Despite its elegant and simple form, the optimal measurement~\eqref{eq:SLD} is, unfortunately, often an experimentally challenging one to perform. Thus, we will explore more feasible measurements in the next sections. It is worth mentioning that the problem of phase estimation by means of Gaussian measurements has been addressed in~\cite{PhysRevA.73.033821,oh2019optimal} for a single mode, while using Gaussian measurements in the Bayesian approach to metrology was considered in~\cite{Morelli_2021, kundu2021machine}.

\subsection{Thermometry using Gaussian measurements}
\label{measu}
Gaussian measurements are an important family of experimentally realisable measurements. They have been defined formally in section~\ref{sec:Gaussian_states_measure}; alternatively, they can be characterized as all those measurements that are realizable by appending ancillas in the vacuum state, applying arbitrary Gaussian unitaries, and performing homodyne detections~\cite{nogo3}. %As we defined them officially in
We can characterize them by means of their covariance matrix ${\bm \sigma}_{s}^M$. Our first observation is that the optimal measurement should have a covariance matrix representing a pure quantum state.

\subsubsection{Observation.---Pure Gaussian measurements are always better}
To prove this claim, we show that any non-pure Gaussian measurement can be cast as classical post-processing of another pure Gaussian measurement, therefore, not increasing its (classical Fisher) information. 

The POVM elements of a Gaussian measurement can be written in the form
\begin{equation} \label{eq:Gaussian_Child}
	{M}_{\bm{a}|{s}}=\int\frac{\bdm{y}}{(2\pi)^{2m}}\,
	D\big(\bm{y}\big)e^{{-\bm{y}^\intercal{\bm{\sigma}_{s}^M}\bm{y}/2-i \bm{a}^\intercal \bm{\Omega}\bm{y}}}.
\end{equation}
If the measurement is not pure, then there exists a pure covariance matrix $\bm{\sigma}_{t}^{M}$ s.t.\ ${\bm{\sigma}_{s}^M} - \bm{\sigma}_{t}^{M}\geqslant 0$. This follows directly from Williamson's theorem~\cite{willy} (see also~\cite[Lemma~5]{LL-log-det} for a direct proof). The POVM measurement elements corresponding to this pure CM are
\begin{equation} \label{eq:Gaussian_parent}
	{M}_{\bm{b}|t}=\int\frac{\bdm{y}}{(2\pi)^{2m}}\,
	D\big(\bm{y}\big)e^{{-\bm{y}^\intercal\bm{\sigma}_{t}^{M}\bm{y}/2-i \bm{b}^\intercal \bm{\Omega}\bm{y}}}.
\end{equation}
If ${M}_{\bm{a}|{s}}$ can be reproduced by ${M}_{\bm{b}|t}$, then we should find the probability measure $\pi({\bm a}|{\bm b})$ s.t., 
\begin{align} \label{eq:Gaussian_postprocess}
	{M}_{\bm{a}|{s}} = \int \frac{\bdm{b}}{(2\pi)^{m}} \pi({\bm a}|{\bm b}) {M}_{\bm{b}|t} ~~\forall {\bm a}.
\end{align}
By combining \eqref{eq:Gaussian_parent} and \eqref{eq:Gaussian_postprocess} we obtain
\begin{align}
	{M}_{\bm{a}|{s}} & = \iint \frac{\bdm{y}}{(2\pi)^{2m}} \frac{\bdm{b}}{(2\pi)^{m}} \pi({\bm a}|{\bm b}) \,
	D\big(\bm{y}\big)\nonumber\\
	&~~\times e^{{-\bm{y}^\intercal\bm{\sigma}_{t}^{M}\bm{y}/2-i \bm{b}^\intercal \bm{\Omega}\bm{y}}}.
\end{align}
If we compare with \eqref{eq:Gaussian_Child}, we should have
\begin{align}
	\int \frac{\bdm{b}}{(2\pi)^{m}} \pi({\bm a}|{\bm b}) e^{-i \bm{b}^\intercal \bm{\Omega}\bm{y}} = e^{{-\bm{y}^\intercal({\bm{\sigma}_{s}^M}-\bm{\sigma}_{t}^{M})\bm{y}/2-i \bm{a}^\intercal \bm{\Omega}\bm{y}}}.
\end{align}
Finally by taking the inverse Fourier transform we find
\begin{align}
	\pi({\bm a}|{\bm b}) & = \frac{e^{\frac{1}{2}({\bm b}-{\bm a})^\intercal{\bm \Omega}({\bm{\sigma}_{s}^M}-\bm{\sigma}_{t}^{M})^{-1}{\bm \Omega}^\intercal({\bm b}-{\bm a})}}{\sqrt{\det ({\bm{\sigma}_{s}^M}-\bm{\sigma}_{t}^{M})}}.
\end{align}
Note that the fact that the Fourier transform exists owes to the fact that ${\bm{\sigma}_{s}^M}-\bm{\sigma}_{t}^{M} \geqslant 0$.

In this manner, we have proven that the optimal Gaussian measurements should be chosen from the pool of pure covariance matrices, which can always be written as ${\bm \sigma}_{s}^M = {\bm S}^M ({\bm S}^M)^\intercal$, with ${\bm S}^M$ being a symplectic transformation.
\subsubsection{Classical Fisher information of Gaussian measurements}
As we saw in section~\ref{sec:Gaussian_states_measure}, when performing a Gaussian measurement on Gaussian systems the resulting probability distribution, given by~\eqref{eq:prob_comp}, is also Gaussian. It %This Gaussian distribution
is charachterized by its displacement vector ${\bm d}-{\bm d}_s^M$ and covariance matrix ${\bm V}\coloneqq {\bm \sigma}_{s}^M+{\bm \sigma}$, with (${\bm d}$, ${\bm \sigma}$) and (${\bm d}_s^M$, ${\bm \sigma}_s^M$) being the (displacement vector, covariance matrix) of the system and the measurement, respectively. 
The Fisher information of classical Gaussian distributions is well known and reads
\begin{align}\label{eq:CFI_Gaussian_main}
    {\cal F}^{\rm C}({\bm d},{\bm \sigma};{\bm \sigma}_s^{M}) & = \partial_{T} {\bm d}^{T}({\bm \sigma} + {\bm \sigma}_s^{M})^{-1}\partial_{T}{\bm d}  \nonumber\\
    &~~+ \frac{1}{2}\Tr \left[\left(({\bm \sigma} + {\bm \sigma}_s^{M})^{-1} \partial_{T}{\bm \sigma}\right)^2 \right],
\end{align}
where we explicitly include in the argument that the classical Fisher information is obtained for measurement ${\bm \sigma}_s^{M}$ conditioned on the covariance matrix of the system being ${\bm\sigma}$. 
In what follows we focus on systems with zero first-order moments, and use the convention ${\cal F}^{\rm C}({\bm \sigma};{\bm \sigma}_s^{M}) \equiv {\cal F}^{\rm C}({\bm 0},{\bm \sigma};{\bm \sigma}_s^{M})$. 
%\tcr{[LL: I guess you mean vanishing first moments, i.e.\ vanishing expectation value of all the operators $x_j,p_k$?]} 
This is motivated by the fact that, at thermal equilibrium, one can always change the basis with a temperature-independent transformation to one with zero displacements. Moreover, in most non-thermal scenarios the quadratic nature of the system-sample Hamiltonian is such that the displacement vector of the system vanishes~\cite{PhysRevLett.122.030403,PhysRevA.96.062103,planella2020bath,Lampo2017bosepolaronas,khan2021sub}.
The proof of Eq.~\eqref{eq:CFI_Gaussian_main} can be found, e.g., in Refs.~\cite{monras2013phase,malago2015information}; in the appendix we also present a simple proof for the case with zero displacement.

\subsubsection{Homodyne detection and equivalence with the error propagation}
Homodyne detection is a projective measurement in the basis of one of the quadratures. For example, consider a single-mode Gaussian system, and assume that we want to measure it in the position basis. The measurement covariance matrix can be written as ${\bm \sigma}_{R_{1}}^{M} \coloneqq \lim_{r\to\infty} {\rm diag} [1/r, r]$. By replacing in \eqref{eq:CFI_Gaussian_main} one finds the classical Fisher Information for the position measurement
\begin{align}\label{eq:CFI_HD}
    {\cal F}^{\rm C}({\bm \sigma};{\bm \sigma}_{R_{1}}^{M}) 
    & =\hspace{-1mm} \lim_{r\to\infty} \frac{1}{2} \Tr \left[\left(({\bm \sigma} +  {\rm diag}[1/r, r])^{-1} \partial_{T}{\bm \sigma}\right)^2 \right]\nonumber\\
    &= \frac{|\partial_{T}{\bm \sigma}_{11}|^2}{2{\bm \sigma}_{11}^2}.
\end{align}
%Many
Often times, the \textit{inverse} of the error propagation is an alternative way of characterising the precision. The error propagation characterizes the estimation error when using some observable $O$ and is defined as
\begin{align}\label{eq:error_prop}
    \delta^2 T (\rho;O)_{\rm EP} \coloneqq \frac{{\rm Var}(O)}{|\partial_{T}\average{O}|^2}.
\end{align}
Regarding position measurement, if we replace the observable with the square of the position operator, i.e., $O\to R_1^2$, then by using Wick's theorem and Eqs.~\eqref{eq:CFI_HD}  and \eqref{eq:error_prop} one finds that
\begin{align}
    \delta^2 T (\rho;R_1^2)_{\rm EP} = \frac{1}{{\cal F}^{\rm C}({\bm \sigma};{\bm \sigma}_{R_{1}}^{M})}.
\end{align}
This equivalence between the error propagation of the squared quadratures and the inverse of classical Fisher information holds for any parameter estimation problem.
Needless to say that, in case we were to measure the momentum, we just need to replace $R_1\to R_2$ and ${\bm \sigma}_{11}\to {\bm \sigma}_{22}$. Note that this result holds regardless of the basis, and in particular for rotated quadrature measurements as well; one just needs to rewrite the covariance matrix of the system ${\bm \sigma}$ in the rotated basis. The equivalence between the error propagation and the inverse of the classical Fisher information does not necessarily hold in multimode scenario. Nonetheless, for the trivial case with a tensor product state in a basis that does not depend on the parameter---where the classical Fisher information becomes additive---the same result holds. That is to say if the system covariance matrix satisfies ${\bm \sigma}=\oplus{\bm \sigma}^k$, and $\partial_T{\bm \sigma}=\oplus_k\partial_T{\bm \sigma}^k$, then any local measurement in the form of ${\bm \sigma}_{s}^M = \oplus_k {\bm \sigma}_{\rm k,HD}^{M}$---where ${\bm \sigma}_{\rm k,HD}^M$ is an arbitrary Homodyne detection on the $k$th mode---has a classical Fisher information that is equivalent to the inverse of the error propagation for $O=\otimes_k \left(R_{\rm k,HD}\right)^2$---where $R_{\rm k,HD}$ is the local quadrature in the Hilbert space of the $k$th mode corresponding to ${\bm \sigma}_{\rm k,HD}^M$. 

Before proceeding to the next section, let us point out that often performing a perfect homodyne detection requires infinite energy. In Appendix \ref{App:HD} we briefly address such a scenario. In particular, we show that there will be a bias in our temperature estimate, which vanishes linearly as one increases the energy in the local oscillator.

\subsubsection{Heterodyne detection} 
A Heterodyne measurement can be represented---in the single-mode case---by the covariance matrix ${\bm \sigma}_{\rm Het}^M = {\rm diag}[1, 1]$, which remains the same under local rotation of the quadratures. One can find the corresponding classical Fisher information by replacing this in Eq.~\eqref{eq:CFI_Gaussian_main}. If the system's covariance matrix and its derivative can be diagonalised in the same basis, i.e., ${\bm \sigma} = {\rm diag}[{\bm \sigma}_{11}, {\bm \sigma}_{22}]$, and $\partial_T{\bm \sigma} = {\rm diag}[\partial_T{\bm \sigma}_{11}, \partial_T{\bm \sigma}_{22}]$ then the classical Fisher information gets a simple form 
\begin{align}
    \hspace{-2mm}{\cal F}^{\rm C}({\bm \sigma};{\bm \sigma}_{\rm Het}^M) = \frac{1}{2}\left(
    \left[\frac{\partial_T{\bm \sigma}_{11}}{1+{\bm \sigma}_{11}}\right]^2 + \left[\frac{\partial_T{\bm \sigma}_{22}}{1+{\bm \sigma}_{22}}\right]^2
    \right).
\end{align}
If the system is at thermal equilibrium, its covariance matrix is diagonal in a basis that is independent of temperature, and therefore the above equation is useful. Otherwise, if one deals with non-thermal states (such as, e.g., steady states that are non-thermal\cite{PhysRevLett.122.030403,PhysRevA.96.062103,planella2020bath,Lampo2017bosepolaronas,khan2021sub} or dynamical non-thermal states~\cite{PhysRevLett.114.220405}), one should avoid using this, and directly use Eq.~\eqref{eq:CFI_Gaussian_main}.
\subsubsection{The optimal Gaussian measurement}
Given a covariance matrix and its derivative, it is a natural question to find the Gaussian measurement with the maximum Fisher information
\begin{align}
    \hspace{-2mm}{\bm \sigma}_{\rm max}^M & \coloneqq {\rm arg}\max_{{\bm \sigma}_s^{M}}
    \frac{1}{2}\Tr \left[\left(({\bm \sigma} + {\bm \sigma}_s^{M})^{-1} \partial_{T}{\bm \sigma}\right)^2 \right],\\
    \hspace{-2mm}{\bm \sigma}_{\rm min}^M & \coloneqq {\rm arg}\min_{{\bm \sigma}_s^{M},{\rm pure}}
    \frac{1}{2}\Tr \left[\left(({\bm \sigma} + {\bm \sigma}_s^{M})^{-1} \partial_{T}{\bm \sigma}\right)^2 \right].
    % \\
    % {\cal F}^{\rm C}({\bm \sigma};{\bm\sigma}_{\rm max}) & \coloneqq \max_{{\bm \sigma}_s^{M}}
    % \frac{1}{2}\Tr \left[\left(({\bm \sigma} + {\bm \sigma}_s^{M})^{-1} \partial_{T}{\bm \sigma}\right)^2 \right],\\
\end{align}
As already advanced, ${\bm \sigma}_{\rm max}^M$ should represent a pure state. We also defined ${\bm \sigma}_{\rm min}^M$ as the covariance matrix of the \textit{pure} measurement with minimum Fisher information.
If the covariance matrix of the system and its derivative are proportional to identity (e.g., if we have a thermal state), then it is easy to find the optimal Gaussian measurement in the single-mode case. Let ${\bm \sigma} = \nu {\bm I}_2$, and $\partial_T{\bm \sigma}={\nu}^{\prime}{\bm I}_2$. Then, we also note that the measurement covariance matrix can be cast as 
${\bm \sigma}_{r}^M={\bm O}{\bm r}^M{\bm O}^{T}$, with ${\bm O}$ being an orthogonal transformation, and ${\bm r}^M={\rm diag}[{ r}, 1/{r}]$ with the squeezing parameter $r\in [0,\infty)$. The classical Fisher information does not depend on the orthogonal rotation ${\bm O}$, and reads
\begin{align}
    {\cal F}^{\rm C}({\bm \sigma};{\bm \sigma}_{r}^M) = \frac{(\nu^{\prime})^2}{2}\left[
    \frac{1}{(\nu+r)^2}+\frac{1}{(\nu+1/r)^2}
    \right].
\end{align}
Note that this does not change by the transformation $r\to 1/r$, and therefore, we can limit the domain to $r\in [1,\infty)$.
We can find the value of $r$ corresponding to the minimum and maximum of the Fisher information by simple algebra, which is summarised in Table~\ref{table:CFI_Gaussian}. For example, the maximum is given by the expression
\begin{align}
    {\cal F}^{\rm C}({\bm \sigma};{\bm \sigma}_{\max}^M)&=\max_{r\geqslant 1} {\cal F}^{\rm C}({\bm \sigma};{\bm \sigma}_{r}^M) \nonumber\\
    &= {\nu'}^2 \max\left\{ \frac{1}{2\nu^2},\, \frac{1}{(\nu+1)^2}\right\}.
\end{align}
It can be seen that the optimal measurement is either Homodyne or Heterodyne detection. The freedom in the orthogonal rotation ${\bm O}$ means that one can perform these measurements in an arbitrary basis. 

If the state is at thermal equilibrium, then $\nu = \coth(\omega/2T)$, with $\omega$ being the energy of the Bosonic mode and $T$ being the temperature. The fact that the optimal Gaussian measurement depends on temperature is interesting; on the contrary, the optimal non-Gaussian measurement is always performed in the basis of the system Hamiltonian, regardless of the temperature. 
Furthermore, our result shows not only that at high temperatures the Heterodyne detection is the best Gaussian measurement, but also that its performance approaches that of the optimal measurement. %\tcr{[LL: really? Perhaps we want to say that its performance approaches that of the optimal measurement, which is quite a different thing.]} 
To see this, note that the quantum Fisher information of a single Bosonic mode (that is a Harmonic oscillator) at thermal equilibrium is given by~\cite{PhysRevLett.114.220405}
\begin{align}\label{eq:QFI_HO}
    {\cal F}^{\rm Q}({\rm HO}) = \frac{\omega^2}{4T^4\sinh^2(\omega/2T)},
\end{align}
and thus at $T\to\infty$ the leading order behaves as ${\cal F}^{\rm Q}({\rm HO}) = 1/T^2$. 
%\tcr{[LL: I'm afraid this doesn't really make sense; if you take the limit on the LHS then on the RHS there should be no $T$.]}. 
Similarly, one can see that at the same asymptotic limit, the leading order of the Fisher information of Heterodyne measurement behaves as $ {\cal F}^{\rm C}({\bm \sigma};{\bm \sigma}_{\rm max}^M) = {\cal F}^{\rm C}({\bm \sigma};{\bm \sigma}_{\rm Het}^M) = 1/T^2$. %\tcr{[LL: same as before.]}

\begin{table*}[t!]
	\begin{center}
	\begin{tabular}{ |c||c|c|c|c|  }
		\hline
		$\text{Consideration}$  & $r_{\min}$ & ${\cal F}^{\rm C}({\bm \sigma};{\bm \sigma}_{\rm min}^{M})$   & $r_{\max}$ & ${\cal F}^{\rm C}({\bm \sigma};{\bm \sigma}_{\rm max}^{M})$ \\
		\hline
		$1\leqslant\nu\leqslant 2$  & $1$ & $\frac{\nu^{\prime 2}}{(1+\nu)^2}$   & $\infty$ & $\frac{\nu^{\prime 2}}{2\nu^2}$ \\
		\hline
		$2\leqslant\nu\leqslant 1+\sqrt{2}$  & $\frac{\nu(\nu^2-3)}{2}$ & $\frac{\nu^{\prime 2}(\nu^2-2)}{2(\nu^2-1)}$   & $\infty$ & $\frac{\nu^{\prime 2}}{2\nu^2}$ \\
		\hline
		$1+\sqrt{2}\leqslant\nu$  & $\frac{\nu(\nu^2-3)}{2}$ & $\frac{\nu^{\prime 2}(\nu^2-2)}{2(\nu^2-1)}$   & 1 & $\frac{\nu^{\prime 2}}{(\nu+1)^2}$ \\
		\hline
	\end{tabular}
\end{center}
\caption{Maximal and minimal values of the CFI ${\cal F}^{\rm C}({\bm \sigma};{\bm \sigma}_{s}^M)$ for Gaussian measurements in one-mode thermal states parametrized by $\nu$, and the measurements that achieve them with the squeezing parameter $r$. Only pure measurements are considered. The optimal measurement is either Homodyne detection (for the low-temperature regime) or Heterodyne detection (for high temperatures). This is contrary to the optimal [non-Gaussian] measurement, which is a projective measurement in the energy basis regardless of the temperature.}
\label{table:CFI_Gaussian}
\end{table*}

\begin{figure*}
	\centering
	\includegraphics[width=.7\linewidth]{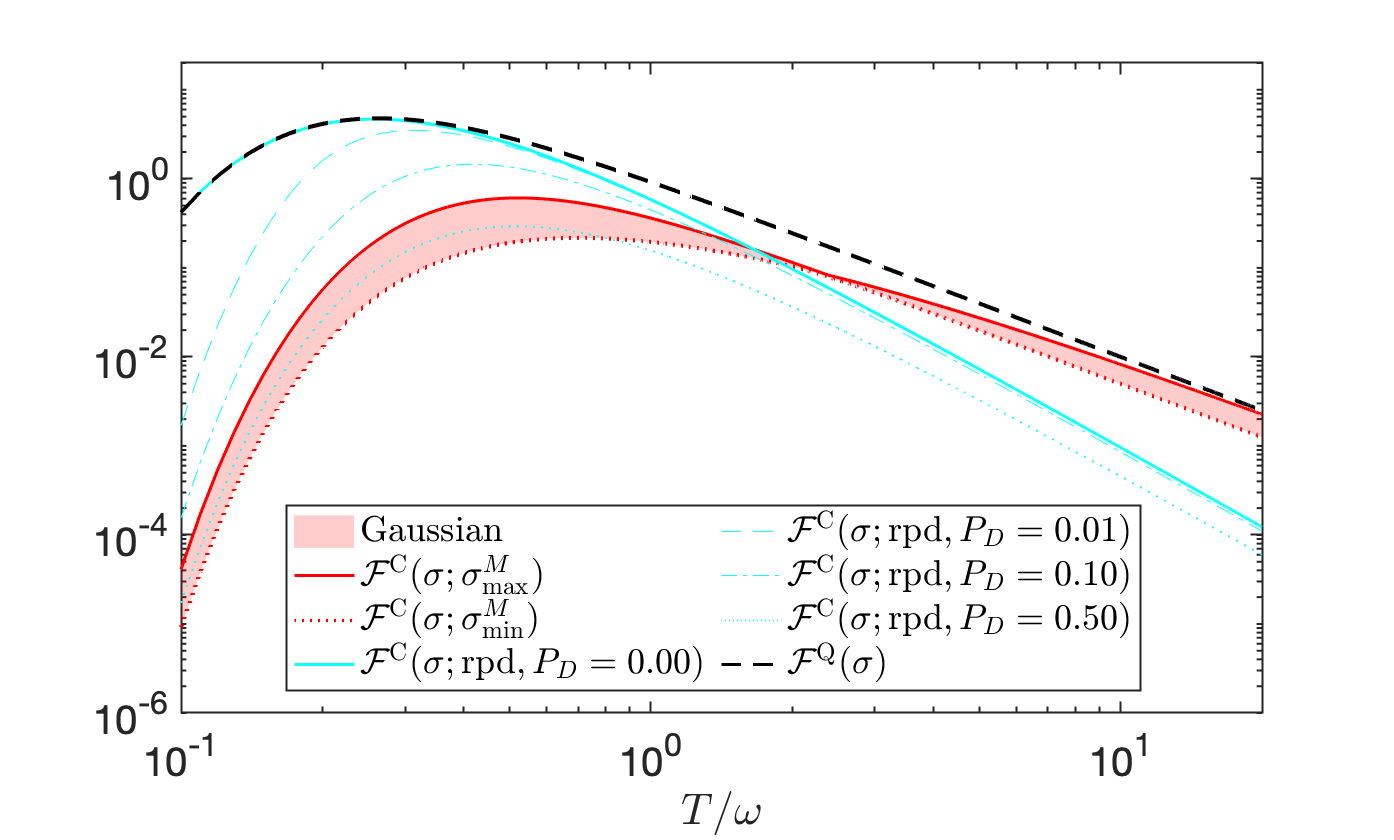}
	\caption{Temperature estimation of single-mode thermal systems at thermal equilibrium. We compare the ultimate limit on thermometry precision posed by the Quantum Fisher Information (dashed black line) and the precision achieved by the optimal Gaussian measurement (solid red), the worst pure Gaussian measurement (dotted red) and single-mode {on/off} photo-detection (blue) versus the temperature $T$. The shaded area represents the precision achievable by the set of single-mode Gaussian measurement. We also show the performance of {on/off} photo-detection including imperfect detectors, characterized by the parameter $P_D$ (dashed and dotted blue lines).}
	\label{fig:QFI_vs_Gaussian_vs_rpd}
\end{figure*}

In Figure~\ref{fig:QFI_vs_Gaussian_vs_rpd} the Fisher information of Gaussian measurements is depicted versus temperature. We compare this to the quantum Fisher information. Unlike in the high-temperature regime, where Gaussian measurements (Heterodyne specifically) are close to optimal, at low temperatures %regime they are extremely bad.
they perform rather poorly.
In Section~\ref{sec:PD} we will see that {on/off} photo-detection measurements are a viable alternative in this latter case, and specifically that they perform close to optimal at low temperatures.

\subsubsection{Optimal Gaussian measurements on multimode systems}
When dealing with a multimode scenario, the problem becomes harder, and we were unable to derive an analytical solution. The main difficulty arises from the fact that the objective function to be optimised, i.e., the Fisher information, is a nonlinear function of the measurement covariance matrix. 
Nonetheless, one can design an algorithm to find the optimal measurement numerically. To begin with, one can use the fact that the covariance matrix of the optimal measurement is always pure %---in all modes---
and therefore belongs to the family of symmetric symplectic transformations, i.e.,  ${\bm \sigma}_{\rm max}^M \in \{\bar{\bm S}^M\coloneqq{\bm S}^M({\bm S}^M)^\intercal|{\bm S}^M{\bm \Omega}({\bm S}^M)^\intercal = {\bm \Omega}\}$, where ${\bar{\bm S}^M}$ is a symmetric symplectic transformation, by definition.
We can then run an optimisation program, e.g., ``fmincon'' on MATLAB programing language, to maximise ${\cal F}^{\rm C}({\bm \sigma};{\bm S}^M({\bm S}^M)^\intercal)$---alternatively to minimise -${\cal F}^{\rm C}({\bm \sigma};{\bm S}^M({\bm S}^M)^\intercal)$---subject to the constraint that ${\bm S}^M{\bm \Omega}({\bm S}^M)^\intercal = {\bm \Omega}$. See the appendix for an example code.

\subsubsection{Multimode Gaussian states at thermal equilibrium: Local versus global Gaussian measurements}
\begin{figure*}[t!]
    \centering
    \includegraphics[width=.8\linewidth]{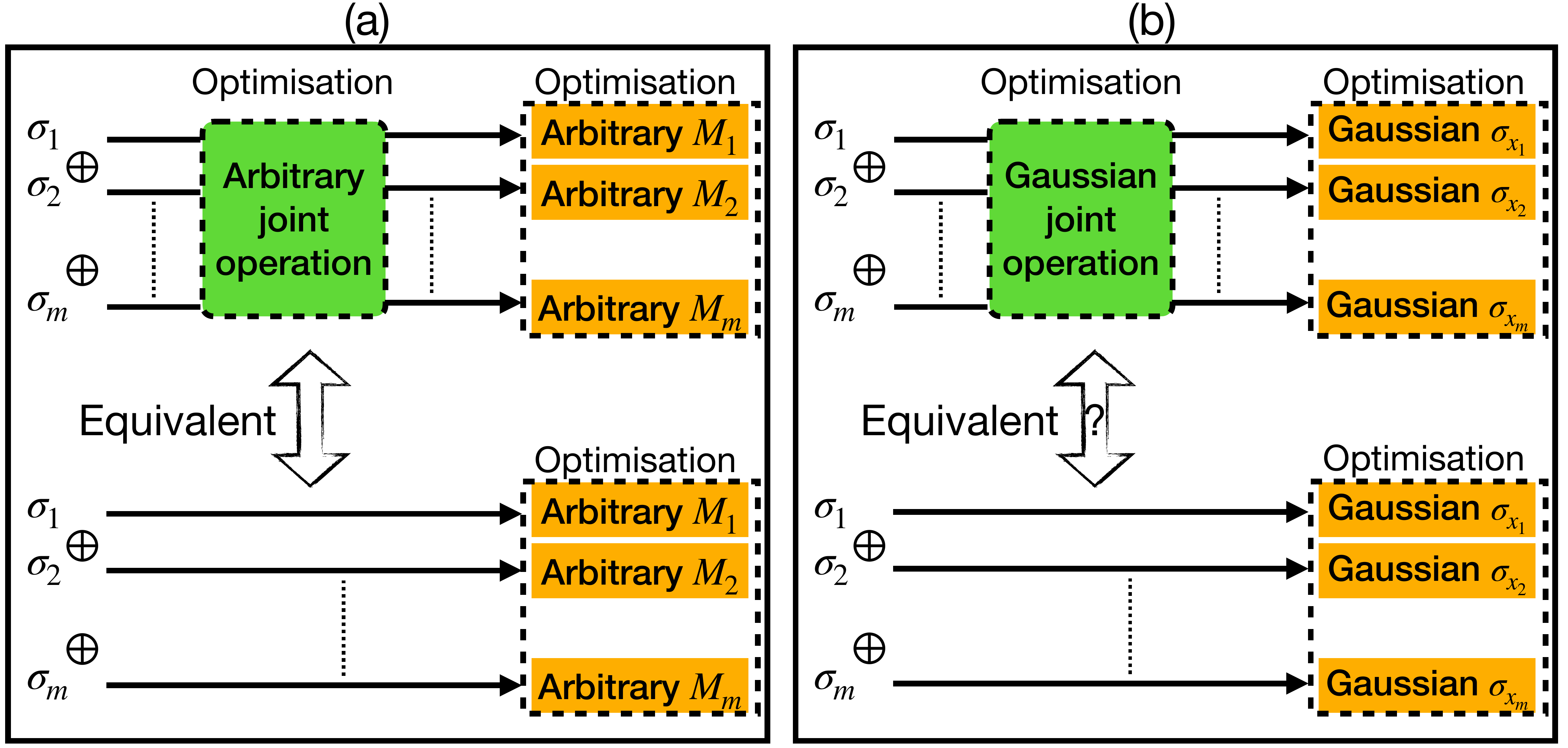}
    \caption{The role of global measurements when the system is in a product state i.e., the covariance matrix is a direct sum. (a) The system undergoes an arbitrary [physical] joint operation then each mode is locally measured. We have no restriction in the performed local measurements. In this case one can show that the optimal scenario is a  local one. That is, we do not need to perform any joint operation; optimisation over arbitrary local measurements is sufficient. This is so because of the additivity of the quantum Fisher information. (b) Same as (a), but when the joint operation and the local measurements are all Gaussian. Unlike (a), if we are restricted to perform only Gaussian measurements, it is not clear whether joint Gaussian operations are beneficial or not. We conjecture that at thermal equilibrium they are not.}
    \label{fig:Global_vs_Local}
\end{figure*}
If the system is at thermal equilibrium, one can always diagonalise its covariance matrix by means of some temperature-independent symplectic transformation---i.e., to write it in the normal-mode basis. The symplectic transformation is a Gaussian process, and can be considered as part of the Gaussian measurement which we will optimise anyway. Therefore, we can focus only on initial thermal states with an uncorrelated covariance matrix ${\bm\sigma} = \oplus_k{\nu_k}{\bm I}_2$. 
As for the measurement, one can first perform a joint Gaussian symplectic transformation on all modes---s.t., ${\bm \sigma}\to {\bm S}_s^M{\bm \sigma} ({\bm S}_s^M)^\intercal$---and then perform a local Gaussian measurement ${\bm \sigma}_{s}^{M} = {\oplus}_k{\bm \sigma}_{s_k}^M$. The idea is to find the optimal ${\bm S}_s^M$ and ${\bm\sigma}_{s_k}^M~\forall k$.

Let us start by a scenario that only allows for local Gaussian measurements, i.e., ${\bm S}^M_s={\bm I}_{2m}$. In this case, one can easily check that the classical Fisher information is additive, that is
\begin{align}
    {\cal F}^{\rm C}(\oplus_k {\bm \sigma}_{k}; \oplus_k{\bm\sigma}_{s_k}^M) = \sum_k {\cal F}^{\rm C}( {\bm \sigma}_{k}; {\bm\sigma}_{s_k}^M).
\end{align}
Therefore, when we optimise over local Gaussian measurements, we have the optimal Fisher information $${\cal F}^{\rm C}(\oplus_k {\bm \sigma}_{k}; \oplus_k{\bm\sigma}_{k,\max}^M) = \sum_k {\cal F}^{\rm C}({\bm \sigma}_{k}; {\bm \sigma}_{k,\max}^M),$$ that is the sum of optimal Fisher information of individual modes. The question that remains to be answered is: can the performance of Gaussian measurements get better if we allow for joint measurements, i.e., if ${\bm S}_s^M \neq {\bm I}_{2m}$?

As depicted in Fig.~\ref{fig:Global_vs_Local}, the global Gaussian measurements can be seen as the operation of some joint Gaussian unitary on all modes---namely an $m$-mode symplectic transformation---followed by local measurements. The question is then whether the joint Gaussian operation can be beneficial. Mathematically, we know that ${\cal F}^{\rm C}(\oplus_k {\bm \sigma}_{k}; {\bm \sigma}^M_{\max}) \geqslant {\cal F}^{\rm C}(\oplus_k {\bm \sigma}_{k}; \oplus_k{\bm \sigma}_{k,\max}^M)$, since the local measurements are a subset of global ones. We now wonder whether this inequality is in fact an equality, i.e., if 
\begin{align}\label{eq:equivalence_conj}
    {\cal F}^{\rm C}(\oplus_k {\bm \sigma}_{k}; {\bm \sigma}^M_{\max})  \overset{?}{=} {\cal F}^{\rm C}(\oplus_k {\bm \sigma}_{k}; \oplus_k{\bm \sigma}_{k,\max}^M).
\end{align}
Note that, if we were not restricted to Gaussian measurements and could perform arbitrary measurements, then the equality would always hold. This is so because in this case the optimal Fisher information would reduce to the quantum Fisher information, which is additive and can be achieved by locally optimal measurements, that is 
\begin{align}
    {\cal F}^{\rm C}(\oplus_k {\bm \sigma}_{k}; M_{\max}) &= {\cal F}^{\rm C}(\oplus_k {\bm \sigma}_{k}; \oplus_k M_{k,\max}) \nonumber\\&= {\cal F}^{\rm Q}(\oplus_k {\bm \sigma}_{k}) = \sum_k {\cal F}^{\rm Q}( {\bm \sigma}_{k}),
\end{align}
with $M_{\max}$ being the optimal joint measurement---that is not necessarily Gaussian---performed on all modes, while $M_{k,\max}$ is the optimal local measurement---again not necessarily Gaussian---performed on mode $k$.

Unfortunately, we cannot prove the equivalence Eq.~\eqref{eq:equivalence_conj} in the general case. However, in the appendix we prove it analytically when all modes have identical covariance matrices---i.e., when ${\bm \sigma_k} = {\bm \sigma}_j \forall j,k$. Furthermore, our numerical simulations support the equivalence for arbitrary product states. Based on these observations we conjecture this is generally true, however, a rigorous proof is missing currently.

%%%%%%%%%%%%%%%%
\subsection{{On/Off} Photo-Detection}\label{sec:PD}
Another type of measurement that will be shown to be relevant for thermometry and which can be of interest for experimental implementation is {on/off photo-detection}. This is a 2-outcome (detected or not detected) measurement and therefore it is not Gaussian. However, a link can be made between {this type of} photon-detection's POVM elements and the POVM elements of a Gaussian measurement, characterized in Section \ref{2}.
In what follows we consider realistic photo-detection, in which we allow for a probability $P_D$ of random counts---see e.g., \cite{PhysRevX.6.021039,BARNETT199845} for the error sources contributing to the random counts.
For a single-mode system the POVM elements can then be written as
\begin{align}\label{povmelements}
{M}_{0| \rm{rpd}} & = (1-P_D)\ket{0}\bra{0}, \nonumber \\
{M}_{{\bar 0}| {\rm rpd}} & = I -   {M}_{0| \rm{rpd}},
\end{align}
where $ \ket{0}\bra{0}$ is merely the vacuum state, and $I$ is the identity operator in the Hilbert space of the quantum system (not to be mistaken with ${\bm I}_2$). The vacuum itself is a pure Gaussian state with a CM equal to identity ${\bm \sigma}_{\rm vac}^M={\bm I}_2$ and null displacement vector. The probability of finding any Gaussian system---with density matrix $\rho$, covariance matrix ${\bm \sigma}$ and a vanishing displacement vector---in vacuum is then given by
\begin{align}\label{eq:P_0_PD}
    P(0|{\bm \sigma};{\rm rpd})  &= \Tr [\rho {M}_{0| {\rm rpd}}] = (1-P_D)\Tr [\rho \ket{0}\bra{0}]
     \nonumber\\
     &= \frac{1-P_D}{\sqrt{{\rm det} (\bm{I}_{\bm{2}}+\bm{\sigma})}},
\end{align}
where we use Eq.~\eqref{eq:Gaussian_Multid_int}. Of course we also find the probablity of detection to be $P({\bar 0}| {\rm rpd};{\bm \sigma}) = 1 -  P(0| {\rm rpd};{\bm \sigma})$. The Classical Fisher Information associated to this probability distribution can be simply found from \eqref{eq:CFI_general} and reads
\begin{align}\label{cfipd}
    {\cal F}^{\rm C}( {\bm \sigma};{\rm rpd})& = \frac{[\partial_{T} P(0|{\bm \sigma}; {\rm rpd})]^2}{P(0|{\bm \sigma}; {\rm rpd})P({\bar 0}|{\bm \sigma}; {\rm rpd})},
\end{align}
As an example, we can use {on/off} photo-detection for thermometry of a single mode at thermal equilibrium, with ${\bm \sigma}=\coth(\omega/2T){\bm I}_2$.
Using \eqref{eq:P_0_PD} one finds $P(0|{\bm \sigma}; {\rm rpd})=(1-P_D)P_0$ with $P_0\coloneqq (1-\exp[-\omega/T])^{-1}$. Thus the Fisher information is given by
\begin{align}
    {\cal F}^{\rm C}({\bm \sigma};  {\rm rpd})= \frac{(1-P_D)[\partial_T P_0]^2}{P_0\left[1-(1-P_D)P_0\right]}.
\end{align}
In Fig.~\ref{fig:QFI_vs_Gaussian_vs_rpd} we plot the classical Fisher information of realistic  {on/off} photo-detection for various $P_D$, and benchmark it against the optimal measurement and the optimal Gaussian measurement. Ideal  {on/off} photo-detection performs close to optimal specifically at low temperatures. 
%In particular, ideal photo-detection (i.e., $P_D=0$) is in fact optimal \added{IS IT? FOR WHICH RANGE? OR ONLY IN THE LIMIT OF ZERO TEMPERATURE?}. 
However, as soon as detection error is introduced and $P_D 
> 0$, 
the optimality at low temperatures is lost. The intuition behind  {this type of measurement} being optimal is as follows: Firstly, we note that at thermal equilibrium the optimal measurement is in the basis of Hamiltonian, i.e., projection in the number basis. Secondly, at very low $T$ the system is basically frozen in its lowest energy eigenstates, namely it is in the vacuum state with some probability of populating the first excited state. As such, the projection in the number basis can be well approximated by ideal  {on/off} photo-detection. 
Since the probability of observing the system in an excited state is already very low, then any random count will hugely impact our statistics and therefore realistic  {on/off} photo-detection will eventually fail to estimate low enough temperatures. This argument can be made rigorous mathematically as well.
At low temperature, one can see
\begin{align}
    {\cal F}^{\rm C}( {\bm \sigma};{\rm rpd})&\approx\begin{cases}
        \frac{\omega^2 e^{-\omega/T}}{T^4},&\hspace{.5cm} P_D=0\\
        \frac{(1-P_D)\omega^2 e^{-2\omega/T}}{P_DT^4},&\hspace{.5cm} P_D\neq 0
    \end{cases}
\end{align}
Note that the exponential term in the case $P_D\neq 0$ vanishes more rapidly, which is why realistic {on/off} photo-detection performs far from optimal for small enough $T$. Moreover, by using Eq.~\eqref{eq:QFI_HO} one can confirm that the leading order of quantum Fisher information at $T\to 0$ has the same value as the idealized detection, i.e., to the leading order $ {\cal F}^{\rm Q}(\bm \sigma) \approx \omega^2 \exp[-\omega/T]/T^4$.
%%%%%%%%%%%%%%%
\subsubsection{Multimode photo-detection}
For general \textit{m}-mode systems, we can perform  {on/off} photo-detection individually on each mode. The $2^{\textit{m}}$ POVM elements can then be written as tensor product $M_{{\bm a}|{\rm rpd}}={M}_{{\bm a}_1|{\rm rpd}}^1\otimes {M}_{{\bm a}_2|{\rm rpd}}^2 \otimes ... \otimes {M}_{{\bm a}_\textit{M}|{\rm rpd}}^{\textit{m}}$ where ${\bm a}_{i} \in \{0,{\bar 0}\} \forall i $. 
In what follows, we use ${\bm 0}$ to denote the vector of all zeros, ${{\bm 0}}_k$ to denote the vector of all zeros except the $k$th element being ${\bar 0}$, ${{\bm 0}}_{k,l}$ to denote the vector of all zeros except elements $k$ and $l$ which are ${\bar 0}$ and so on.
Given an $m$-mode Gaussian system $\rho$ with the covariance matrix ${\bm \sigma}$, the probability of finding, for example, all modes in the vacuum state is (recall that all first moments of the system vanish)
%\tcr{[LL: I guess you are assuming that the displacement vector of the Gaussian state is $0$, right? Because otherwise that vector has to enter the formula below...]}
\begin{align}\label{p0}
    P({\bm 0}|{\rm rpd};{\bm \sigma}) = \Tr [\rho {M}_{{\bm 0}|{\rm rpd}}]
    & =  \frac{\Pi_{i} (1-P_D^i)}{\sqrt{{\rm det} (\bm{I}_{{2m}}+\bm{\sigma})}},
\end{align}
where $P_D^i$ is the probability of random count for mode $i$, and $\bm{I}_{{2m}}$ is the $2m \times 2m$ identity matrix. Similarly, by using $M_{{\bm 0}_k|{\rm rpd}} = M_{0|{\rm rpd}}^1\dots\otimes (I-M_{0|{\rm rpd}}^k)\dots \otimes M_{0|{\rm rpd}}^m$, the probability of finding the $k$th mode in an excited state while all the other modes at vacuum reads 
% We can also compute the probability of finding the system in an excited state. Lets say mode $1 \leqslant k \leqslant M$ is in an excited state while all other modes are in the vacuum state. Then the associated POVM element will be denoted by ${M}_{0|{\rm rpd}}^1\otimes {M}_{0|{\rm rpd}}^2 \otimes ...\otimes (\bm{I}_{2}- {M}_{0|{\rm rpd}}^k) \otimes ... \otimes {M}_{0|{\rm rpd}}^{\textit{M}}$, which can be written as the sum of two terms, ${M}_{0|{\rm rpd}}^1\otimes {M}_{0|{\rm rpd}}^2 \otimes ...\otimes \bm{I}_{2} \otimes ... \otimes {M}_{0|{\rm rpd}}^{\textit{M}}$ and $- {M}_{0|{\rm rpd}}^1\otimes {M}_{0|{\rm rpd}}^2 \otimes ...\otimes  {M}_{0|{\rm rpd}}^k \otimes ... \otimes {M}_{0|{\rm rpd}}^{\textit{M}}$. The second term is just $-P_{\bm{00...0}}$, which can be computed via Eq.~\ref{p0}. The first term can be thought of as the probability of measuring all modes of the reduced state $tr_{k}[\rho]$ in the vacuum state. Therefore it can be written as:
\begin{align}\label{pk}
     P({{\bm{0}}_k}|{\rm rpd};{\bm \sigma}) = \frac{\Pi_{i\neq k}(1-P_D^i)}{\sqrt{{\rm det}( \tr _{k}[ \bm{I}_{{2m}}+\bm{\sigma}])}} - P({{\bm{0}}}|{\rm rpd};{\bm \sigma}),
\end{align}
where, for any covariance matrix ${\bm \sigma}$, the operation $\tr _{k}[{\bm \sigma}]$ is understood as eliminating the $2$ rows and columns associated to the mode $k$. Following similar steps one can find the probability of other outcomes as well. However, one can restrict to a simplified {version of on/off} photo-detection with $m+2$ POVM elements, $M_{{\rm I}|{\rm rpd}}\coloneqq M_{{\bm 0}|{\rm rpd}}$, $M_{{\rm II}_k|{\rm rpd}}\coloneqq M_{{\bm 0}_k|{\rm rpd}}, \forall k$ and $M_{{\rm III}|{\rm rpd}}\coloneqq I - M_{{\rm I}|{\rm rpd}} - \sum_k M_{{\rm II}_k|{\rm rpd}}$. One can consider yet simpler version, in which one does not distinguish what mode is excited therefore we have only three POVM elements with $ M_{{\rm II}|{\rm rpd}}=\sum_k M_{{\rm II}_k|{\rm rpd}}$ while the other two elements remain the same. These simplified versions should be fairly comparable to the non-simplified version at low temperatures, because at these temperatures the other POVM elements should not play a significant role anyway.

\section{Conclusions}
The problem of thermometry of Gaussian states with experimentally feasible measurements was addressed. The classical Fisher information of Gaussian measurements was charachterized. Finding the best Gaussian measurement can be expressed as an optimisation problem in which one looks for a pure covariance matrix that maximises the classical Fisher information \eqref{eq:CFI_Gaussian_main}. This may be done numerically in general. For thermal states one may find analytical solutions. In particular, for single-mode states at thermal equilibrium the optimal Gaussian measurement was found analytically to be either Homodyne or Heterodyne detection depending on the temperature.  This is unlike the optimal non-Gaussian measurement, which is always a projection in the eigenbasis of the Hamiltonian of the system, regardless of the temperature. For a multimode scenario at thermal equilibrium we 
%\deleted{have speculated what the optimal Gaussian measurement is; our} 
conjecture that the optimal Gaussian measurement is composed of a Gaussian unitary that brings the system into its normal modes basis, followed by a local optimal measurement on each of the normal modes. The proof of this is equivalent to proving Eq.~\eqref{eq:equivalence_conj}, which remains an open question. Our numerical results support this conjecture, and we have proven this analytically for a scenario with all normal modes being equal. 

We further compare the performance of the Gaussian measurements as well as photo-detection-like measurements with the quantum Fisher information, for systems at thermal equilibrium. We see that at low temperatures {on/off} photo-detection-like measurements perform almost optimally, whereas at high temperatures Gaussian measurements are close to optimal and saturate the quantum Cram\'er-Rao bound. We expect that for systems that are out of thermal equilibrium these results do not hold. For example, it is shown in \cite{PhysRevA.96.062103,khan2021sub,PhysRevLett.122.030403} that at ultracold temperatures and for non-equilibrium systems, sometimes position measurement---which is a Gaussian measurement---is optimal. A deeper investigation in this direction will be the be subject of future works. 

Lastly, it would be interesting to consider the problem of thermometry with Gaussian measurements in the context of the Bayesian formalism, which has been the subject of few recent works~\cite{rubio2020global,mehboudi2021fundamental,jorgensen2021bayesian,alves2021bayesian}. In particular, since the optimal Gaussian measurement is temperature dependent, this adds an extra challenge into designing optimal thermometry protocols when initial uncertainty about the true temperature is non negligible.

\section*{Acknowledgements}
We are thankful to S.\ Rahimi-Keshari, F.\ Bakhshinezhad, J.\ Kolodynski, P.\ Sekatski, D.\ Rusca, and H.\ Zbinden for constructive discussions. This work was supported by Swiss National Science Foundation NCCR SwissMAP, the Government of Spain (FIS2020-TRANQI, ConTrAct FIS2017-83709-R, and Severo Ochoa CEX2019-000910-S), Fundacio Cellex, Fundacio Mir-Puig, Generalitat de Catalunya (CERCA, AGAUR SGR 1381 and QuantumCAT), the ERC AdG CERQUTE, the AXA Chair in Quantum Information Science. 
This project has received funding from the ``Presidencia de la Agencia Estatal de Investigaci\'on'' within the ``Convocatoria de tramitaci\'on anticipada, correspondiente al a\~no 2020, de las ayudas para contratos predoctorales (Ref. PRE2020-094374 PhD Fellowship) para la formaci\'on de doctores contemplada en el Subprograma Estatal de Formaci\'on del Programa Estatal de Promoci\'on del Talento y su Empleabilidad en I+D+i, en el marco del Plan Estatal de Investigaci\'on Cient\'ifica y Técnica y de Innovaci\'on 2017-2020,
cofinanciado por el Fondo Social Europeo''.
LL acknowledges support from Alexander von Humboldt Stiftung/Foundation. 
%-------------------------
%-------------------------
\onecolumngrid
\appendix
\section{More about Gaussian measurements}
In this Appendix we present some of the details of the toolboxes that were used in the main text. Please note that these results are only presented for self-containment and are not our original contributions.
\subsection{General form of Gaussian measurements}
In general, Gaussian measurements can be performed over an output state of any  Gaussian map on the system $\rho_{s}^M$, which is in general not a pure state, and will then be represented by more general POVM elements that read~\cite{holevo2019quantum,Kiukas_2013,PhysRevA.49.1567}
\begin{equation} 
	\label{eq:General_Gaussian_measurement}
		{M}_{\bm{a}|{s}}=\int\frac{\bdm{y}}{(2\pi)^{2m}}\,
		D\big(\bm{y}\bm{T}_{s}\big)e^{{-\bm{y}^\intercal\bm{N}_{s}\bm{y}^\intercal/2-i \bm{a} \bm{\Omega}\bm{y}}},
\end{equation}
where the complete positivity condition $\bm{N}_{s}-i\bm{T}_{s}\bm{\Omega}\bm{T}_{s}^\intercal/2\geqslant 0$ must hold. The outcome probabilities of this Gaussian measurement are then given by
\begin{align}
	p(\bm{a}|{s},\rho) & = \tr [\rho {M}_{\bm{a}|{s}}] = \int\frac{\bdm{y}}{(2\pi)^{2m}}\,
	\tr [D\big(\bm{y}\bm{T}_{s}\big)\rho]e^{{-\bm{y}^\intercal\bm{N}_{s}\bm{y}/2-i \bm{a}^\intercal \bm{\Omega}\bm{y}}}\nonumber\\
	& = \int\frac{\bdm{y}}{(2\pi)^{2m}}\,
	e^{{\bm{y}^\intercal\bm{T}_{s}\bm{\Omega}\bm{\sigma}\bm{\Omega}\bm{T}_{s}^\intercal\bm{y}/2} + i\bm{d}^\intercal\bm{\Omega}\bm{T}_{s}^\intercal\bm{y}}e^{{-\bm{y}^\intercal\bm{N}_{s}\bm{y}/2-i \bm{a}^\intercal \bm{\Omega}\bm{y}}}\\
	& = \frac{1}{(2\pi)^{2m}}\int d^{2m}\bm{y} e^{\frac{1}{2}\bm{y}^\intercal{\tilde{\bm{\sigma}}} \bm{y} + i {\tilde {\bm{d}}}^\intercal \bm{y} }.
\end{align}
where we have defined $\tilde{\bm{\sigma}} \coloneqq \bm{T}_{s}\bm{\Omega}\bm{\sigma}\bm{\Omega}\bm{T}_{s}^\intercal + \bm{N}_{s} $ and $\tilde {\bm{d}}^\intercal \coloneqq \bm{d}^\intercal\bm{\Omega}\bm{T}_{s}^\intercal - \bm{a}^\intercal \bm{\Omega}$. Integration gives
\begin{align}
	p(\bm{a}|{s},\rho) & = \frac{1}{(2\pi)^{m}\sqrt{\det{\tilde{\bm{\sigma}}}}} e^{\frac{1}{2}\tilde{\bm{d}}^\intercal{\tilde{\bm{\sigma}}}^{-1} \tilde{\bm{d}}}.
\end{align}
The previous scenario is revived by choosing $\bm{N}_{s}=\bm{\Omega}{\bm{\sigma}_{s}^M}\bm{\Omega}$ and $\bm{T}_{s} = \bm{I}_{2m}$.
%--------------
%--------------
\subsection{Local Gaussian measurements, reduced states and covariance matrices}
Suppose we have a multi-mode state and we perform a Gaussian measurement on a subsystem. What is the post measurement state of the remaining modes?
In principle, after performing some POVM on the system and observing an outcome associated with the element ${M}_{\bm{a}|{s}}$, the post measurement state is---up to a unitary---described by the density matrix
\begin{align}
\rho({\bm a}|s) = \frac{\sqrt{{M}_{\bm{a}|{s}}} \rho \sqrt{{M}_{\bm{a}|{s}}}}{\tr  [\sqrt{{M}_{\bm{a}|{s}}} \rho \sqrt{{M}_{\bm{a}|{s}}}] } = \frac{\sqrt{{M}_{\bm{a}|{s}}} \rho \sqrt{{M}_{\bm{a}|{s}}}}{P(\bm{a}|{s},\rho)},
\end{align} 
% with $\sqrt{{M}_{\bm{a}|{s}}} = D(\bm{a}) \rho_{\bm{x}}^{1/2} D^{\dagger}(\bm{a})/\sqrt{(2\pi)^{m}}$, where $M$ is the number of modes of the system. 
Despite the unitary freedom in the post measurement state, the reduced state of the modes over which \textit{we did not} perform the measurement is independent of this unitary.

Suppose we have a two-party density operator, $\rho_{AB}$, of $n+m$ modes, where party $A$ has access to the first $m$ modes and party $B$ has the remaining $n$ modes. A POVM measurement $\{{M}_{\bm{b}|{s}}\}$ is performed on part $B$. That is, the overall POVM elements are given by $\{I_A \otimes {M}_{\bm{b}|{s}}\}$.
% \begin{align}
% \rho_{\rm pm} = \frac{\bm{I}_A \otimes U.\bm{I}_A \otimes \sqrt{{M}_{\bm{a}|{s}}}~\rho_{AB}~\bm{I}_A \otimes \sqrt{{M}_{\bm{a}|{s}}}.\bm{I}_A \otimes U^{\dagger} } {P(\bm{a}|{s},\rho)}.
% \end{align}
Given that we observe the result ${\bm b}$, the post-measurement [reduced] state of party $A$ reads
\begin{align}
\rho_{A}({\bm b},s) \coloneqq \tr _B [\rho_{AB}({\bm b}|s)] = \frac{\tr _B[I_A \otimes {M}_{\bm{b}|{s}} ~ \rho_{AB}]}{P(\bm{b}|{s},\rho_{AB})}.
\end{align}
We are interested in knowing the covariance matrix of this reduced state. To this aim, we need to find the characteristic function, from which one can identify the covariance matrix.
\subsection*{Characteristic function of the reduced density matrix}
By definition, we have
\begin{align}
\chi(\bm{y}|\rho_{A}({\bm b},s) ) = \tr _A[D_A(\bm{y}) \rho_{A}({\bm b},s) ] 
= \frac{\tr [ D_A(\bm{y}) \otimes {M}_{\bm{b}|{s}} ~ \rho_{AB}]}{P(\bm{b}|{s},\rho)}.
\end{align}
Let us focus on the {numerator} since the denominator is only a normalisation factor independent of ${\bm y}$. 
By noting that $D_B({\bm 0}) = I_B$, this can be written as
\begin{align}
\tr [ D_A(\bm{y}) \otimes {M}_{\bm{b}|{s}} ~ \rho_{AB}] & = \tr [D_A(\bm{y})\otimes D_B({\bm 0}) ~ I_A\otimes {M}_{\bm{b}|{s}} ~ \rho_{AB}] \eqqcolon \tr [D_{AB}(\bm{y},{\bm 0}) ~ \varrho_{AB} ~ \rho_{AB}],
\end{align}
where we denote $D_{AB}(\bm{y},{\bm z}) = D_A(\bm{y})\otimes D_B({\bm z})$, and $\varrho_{AB} \coloneqq I_A\otimes {M}_{\bm{b}|{s}}$.
We can expand both $\varrho_{AB}$ and $\rho_{AB}$ in terms of the displacement operator by using Eq.~\eqref{eq:spqd_1}, such that
\begin{align}
\varrho_{AB}\rho_{AB} & = \frac{1}{(2\pi)^{n+m}} \iint d^{2(n+m)}\bm{x} d^{2(n+m)}\bm{z}~  \tr [\varrho_{AB} D(\bm{x})] \tr [\rho_{AB} D(\bm{z})] D^{\dagger}(\bm{x})D^{\dagger}(\bm{z})\nonumber\\
& = \frac{1}{(2\pi)^{n+m}} \iint d^{2(n+m)}\bm{x} d^{2(n+m)}\bm{z} ~ \chi(\bm{x}|\varrho_{AB}) \chi(\bm{z}|\rho_{AB}) D^{\dagger}(\bm{x})D^{\dagger}(\bm{z}),
\end{align}
where $n$ is the number of modes over which we performed the measurement.
Thus we can find the characteristic function of $\varrho_{AB}\rho_{AB}$ using the above expression
\begin{align}
\tr [D_{AB}(\bm{y}) \varrho_{AB}\rho_{AB}]& = \frac{1}{(2\pi)^{n+m}}\iint d^{2(n+m)}\bm{x} d^{2(n+m)}\bm{z} ~ \chi(\bm{x}|\varrho_{AB}) \chi(\bm{z}|\rho_{AB})\tr [D(\bm{y})  D^{\dagger}(\bm{x})D^{\dagger}(\bm{z})]\nonumber\\
& = \frac{1}{(2\pi)^{n+m}}\iint d^{2(n+m)}\bm{x} d^{2(n+m)}\bm{z} ~ \chi(\bm{x}|\varrho_{AB}) \chi(\bm{z}|\rho_{AB}) \nonumber\\
&~~ \times e^{-i\bm{x}^\intercal\bm{\Omega}\bm{z}/2} e^{i\bm{y}^\intercal\bm{\Omega}(\bm{z}+\bm{x})/2} ~ \tr [D(\bm{y}-\bm{x}-\bm{z})]
\nonumber\\
& = \iint d^{2(n+m)}\bm{x} d^{2(n+m)}\bm{z} ~ \chi(\bm{x}|\varrho_{AB}) \chi(\bm{z}|\rho_{AB}) e^{-i\bm{x}^\intercal\bm{\Omega}\bm{z}/2} ~ \delta^{2(n+m)}(\bm{y}-\bm{x}-\bm{z})\nonumber\\
& = \int d^{2(n+m)}\bm{x}~ \chi(\bm{x}|\varrho_{AB}) \chi(\bm{y}-\bm{x}|\rho_{AB}) e^{-i\bm{x}^\intercal\bm{\Omega}\bm{y}/2},
\end{align}
where we used that $\tr [D(\bm a)]=(2\pi)^{n+m}\delta^{2(n+m)}(\bm{a})$.
Let us denote with $\{\bm{\sigma}_{AB}, \bm{d}_{AB}\}$ and $\{{\bm{\sigma}_{s}^M}, \bm{d}_{s}^{M}\}$ the covariance matrix and the displacement vectors of $\rho_{AB} $ and $\varrho_{AB}$, respectively. Therefore
\begin{align}
& \tr [D_{AB}(\bm{y}) \varrho_{AB}\rho_{AB}]\\
& = \int d^{2(n+m)}\bm{x} e^{\frac{1}{2}\bm{x}^\intercal\bm{\Omega} {\bm{\sigma}_{s}^M} \bm{\Omega} \bm{x} + i (\bm{d}_{s}^{M})^\intercal\bm{\Omega} \bm{x} } e^{\frac{1}{2}(\bm{y}-\bm{x})^\intercal\bm{\Omega} \bm{\sigma}_{AB} \bm{\Omega} (\bm{y}-\bm{x}) + i \bm{d}_{AB}^\intercal\bm{\Omega} (\bm{y}-\bm{x}) } e^{-i\bm{x}^\intercal\bm{\Omega}\bm{y}/2}\nonumber\\
& = e^{\frac{1}{2}\bm{y}^\intercal\bm{\Omega} \bm{\sigma}_{AB} \bm{\Omega}\bm{y} + i \bm{d}_{AB}^\intercal\bm{\Omega} \bm{y}}
\int d^{2(n+m)}\bm{x} e^{\frac{1}{2}\bm{x}^\intercal\bm{\Omega} (\bm{\sigma}_{AB} + {\bm \sigma}_{s}^M) \bm{\Omega} \bm{x} + \left(i(\bm{d}_{s}^{M})^\intercal - i\bm{d}_{AB}^\intercal + i\bm{y}^\intercal/2 - \bm{y}^\intercal\bm{\Omega}\bm{\sigma}_{AB}\right)\bm{\Omega} \bm{x} }\nonumber\\
& = e^{\frac{1}{2}\bm{y}^\intercal\bm{\Omega} \bm{\sigma}_{AB} \bm{\Omega}\bm{y} + i \bm{d}_{AB}^\intercal\bm{\Omega} \bm{y} }
\int d^{2(n+m)}\bm{x} e^{\frac{1}{2}\bm{x}^\intercal\bm{\tilde{\sigma}}\bm{x}} e^{i\bm{\tilde{d}}^\intercal \bm{x} },
\end{align}
with $\bm{\tilde{\sigma}} \coloneqq \bm{\Omega} (\bm{\sigma}_{AB} + {{\bm{\sigma}_{s}^M}}) \bm{\Omega}$, and $\bm{\tilde{d}}^\intercal \coloneqq \left((\bm{d}_{s}^{M})^\intercal - \bm{d}_{AB}^\intercal + \bm{y}^\intercal/2 +i \bm{x}^\intercal\bm{\Omega}\bm{\sigma}_{AB}\right)\bm{\Omega} $. 
The integral can also be evaluated using \eqref{eq:Gaussian_Multid_int}, which leads to
\begin{align}
\int d^{2(n+m)}\bm{x} e^{\frac{1}{2}\bm{x}^\intercal\bm{\tilde{\sigma}}\bm{x}} e^{i\bm{\tilde{d}}^\intercal \bm{x} } 
& = \sqrt{\frac{(2\pi)^{2(n+m)}}{{\rm det} {\tilde {\bm{\sigma}}}}} e^{\frac{1}{2}\bm{\tilde{d}}^\intercal{\bm{\tilde{\sigma}}}^{-1} \bm{\tilde{d}}}.
\end{align}
For simplicity let us assume that $\bm{d}_{AB} = \bm{d}_{s}^{M} = 0$. Then we have
\begin{align}
\tr [D(\bm{y}) \varrho_{AB}\rho_{AB}]
& \propto 
e^{\frac{1}{2}\bm{y}^\intercal\bm{\Omega} \bm{\sigma}_{AB} \bm{\Omega}\bm{y}}
e^{\frac{1}{2}\bm{\tilde{d}}^\intercal{\bm{\tilde{\sigma}}}^{-1} \bm{\tilde{d}}}\nonumber\\
& = e^{\frac{1}{2}\bm{y}^\intercal\bm{\Omega} \bm{\sigma}_{AB} \bm{\Omega}\bm{y}}
e^{-\frac{1}{2}\bm{y}^\intercal(\bm{I}/2 + i\bm{\Omega}\bm{\sigma}_{AB})(\bm{\sigma}_{AB} + {{\bm{\sigma}_{s}^M}})^{-1}(\bm{I}/2 + i\bm{\Omega}\bm{\sigma}_{AB})^\intercal\bm{y}}\nonumber\\
& = \exp \left\{{\frac{1}{2}\bm{y}^\intercal\left[\bm{\Omega} \bm{\sigma}_{AB} \bm{\Omega} - (\bm{I}/2 + i\bm{\Omega}\bm{\sigma}_{AB})(\bm{\sigma}_{AB} + {\bm\sigma}_{s}^M)^{-1}(\bm{I}/2 + i\bm{\Omega}\bm{\sigma}_{AB})^\intercal\right]\bm{y}}\right\}.
\end{align}
We can now simply find the post-measurement covariance matrix of party $A$ after the Gaussian measurement ${M}_{\bm{a}|{s}}$ of party $B$ by setting $D_{AB}(\bm{y})\to D_{AB}(\bm{y}_A,{\bm 0}_B)$ in the exponent of the above relation, and by noticing that the covariance matrix of the measurement operator can be written as ${\bm{\sigma}_{s}^M} = \lim_{r \to\infty}r ~ {\bm I}_A \oplus \bm{\sigma}_{s,B}$.
Using the fact that
\begin{align}
(\bm{\sigma}_{AB} + {\bm{\sigma}_{s}^M})^{-1} = \lim_{r\to\infty} (\bm{\sigma}_{AB} + r ~ {\bm I}_A \oplus \bm{\sigma}_{s,B})^{-1} = {\bm 0}_A \oplus (\bm{\sigma}_B + \bm{\sigma}_{s,B})^{-1},
\end{align}
we have
\begin{align}
& \tr [D_{AB}(\bm{y}_A,{\bm 0}_B) \varrho_{AB}\rho_{AB}]\nonumber\\
& \propto 
\exp \left\{{\frac{1}{2}(\bm{y}_A,{\bm 0}_B)^\intercal\left[\bm{\Omega} \bm{\sigma}_{AB} \bm{\Omega} - (\bm{I}/2 + i\bm{\Omega}\bm{\sigma}_{AB})(\bm{\sigma}_{AB} + {\bm{\sigma}_{s}^M})^{-1}(\bm{I}/2 + i\bm{\Omega}\bm{\sigma}_{AB})^\intercal\right](\bm{y}_A,{\bm 0}_B)}\right\}\nonumber\\
& = \exp \left\{{\frac{1}{2}(\bm{y}_A,{\bm 0}_B)^\intercal\left[\bm{\Omega} \bm{\sigma}_{AB} \bm{\Omega} - (\bm{I}/2 + i\bm{\Omega}\bm{\sigma}_{AB})({\bm 0}_A \oplus (\bm{\sigma}_B + \bm{\sigma}_{s,B})^{-1})(\bm{I}/2 + i\bm{\Omega}\bm{\sigma}_{AB})^\intercal\right](\bm{y}_A,{\bm 0}_B)}\right\}\nonumber\\
& = \exp \left\{\frac{1}{2}\bm{y}_A^\intercal\bm{\Omega}_A\left[\bm{\sigma}_{A} - \bm{\sigma}_{AB}^{\rm Corr}(\bm{\sigma}_B + \bm{\sigma}_{s,B})^{-1}(\bm{\sigma}_{AB}^{\rm Corr})^\intercal\right] \bm{\Omega}_A\bm{y}_A\right\}
\end{align}
in which $\bm{\sigma}_{AB}^{\rm Corr}$ is the correlation block of the covariance matrix $\bm{\sigma}_{AB}$, such that
\begin{align}
\bm{\sigma}_{AB} = \left[
\begin{array}{cc}
\bm{\sigma}_{A} & \bm{\sigma}_{AB}^{\rm corr}\\
(\bm{\sigma}_{AB}^{\rm corr})^\intercal & \bm{\sigma}_{B}\\
\end{array}
\right].
\end{align}
Thus, the covariance matrix of $A$, after $B$'s measurement is independent of the outcome and reads
\begin{align}
\bm{\sigma}_{A|s} = \bm{\sigma}_{A} - \bm{\sigma}_{AB}^{\rm Corr}(\bm{\sigma}_B + \bm{\sigma}_{s,B})^{-1}(\bm{\sigma}_{AB}^{\rm Corr})^\intercal.
\end{align}
One can check that the same result holds for $\bm{d}_{AB} \neq \bm{d}_{s}^{M} \neq 0$. See Refs.~\cite{nogo3,doi:10.1080/00107514.2015.1125624} for earlier derivations.
%--------------
%--------------
\section{The classical Fisher information of Gaussian probability distributions}
Given a Gaussian probability distribution $p({\bm a}|\theta)$ which can be written as
\begin{align}
	p(\bm{a}|\theta) & = \frac{e^{-\frac{1}{2}\bm{a}^\intercal{\bm{\sigma}}^{-1} \bm{a}}}{(2\pi)^{M}\sqrt{\det{\bm{\sigma}}}},
\end{align}
we want to find the classical Fisher information. Note that we are assuming that the displacement vector is zero. Similar results are expected if the displacement vector is non zero, but is independent of the parameter. If neither of these are the case, one should use a more compete version of the Fisher information that is given in Eq.~\eqref{eq:CFI_Gaussian_main}---see also refs.~\cite{monras2013phase,malago2015information}. Let us define ${\tilde{\bm \sigma}}\coloneqq {\bm{\sigma}}^{-1}$.
Firstly, note that
\begin{align}
	\partial_{\theta} \log p(\bm{a}|\theta) = -\frac{\partial_{\theta} \det {\bm{\sigma}} }{2\det {\bm{\sigma}}} - \frac{1}{2}\bm{a}_k\bm{a}_l\partial_{\theta}{\tilde{\bm \sigma}}_{kl},
\end{align}
where summation over repeated indices is understood. Recall the Jacobi's formula for the derivative of determinant of [symmetric] matrices
\begin{align}
	\partial_{\theta} \det {\bm A} & = \Tr [{\bm A}^{-1} \partial_{\theta} {\bm A}] \det {\bm A}, \\
	\partial_{{\bm A}_{kl}} \det {\bm A} & = {\bm A}^{-1}_{kl}\det {\bm A}.
\end{align}
By substitution in the expression of classical Fisher information one gets
\begin{align}
	F & = \int d^{2M}{\bm a} p(\bm{a}|\theta)[\partial_{\theta} \log p(\bm{a}|\theta)]^2 
	= \int d^{2M}{\bm a} p(\bm{a}|\theta) \left[-\frac{1}{2}\Tr [{\tilde{\bm \sigma}}\partial_{\theta}{\bm \sigma}] - \frac{1}{2}{\bm a}_k{\bm a}_l\partial_{\theta}{\tilde{\bm \sigma}}_{kl}\right]^2\nonumber\\
	& = \frac{1}{4} \Tr [{\tilde{\bm \sigma}}\partial_{\theta}{\bm \sigma}]^2 
	+ \frac{1}{2}\Tr [{\tilde{\bm \sigma}}\partial_{\theta}{\bm \sigma}] \partial_{\theta}{\tilde{\bm \sigma}}_{kl}\int d^{2M}{\bm a} ~{\bm a}_k{\bm a}_l p(\bm{a}|\theta) 
	+ \frac{1}{4} \partial_{\theta}{\tilde{\bm \sigma}}_{kl}\partial_{\theta}{\tilde{\bm \sigma}}_{mn} \int d^{2M}{\bm a}~{\bm a}_k{\bm a}_l  {\bm a}_m{\bm a}_np(\bm{a}|\theta)\nonumber\\
	& = \frac{1}{4} \Tr [{\tilde{\bm \sigma}}\partial_{\theta}{\bm \sigma}]^2 
	+ \frac{\Tr [{\tilde{\bm \sigma}}\partial_{\theta}{\bm \sigma}] \partial_{\theta}{\tilde{\bm \sigma}}_{kl}}{2(2\pi)^{M}\sqrt{\det{\bm{\sigma}}}}(-2\partial_{\tilde{\bm \sigma}_{kl}})\int d^{2M}{\bm a} ~ e^{-\frac{1}{2}\bm{a}^\intercal{\tilde{\bm \sigma}} \bm{a}}\nonumber\\
	&+ \frac{\partial_{\theta}{\tilde{\bm \sigma}}_{kl}\partial_{\theta}{\tilde{\bm \sigma}}_{mn}}{4(2\pi)^{M}\sqrt{\det{\bm{\sigma}}}}(4\partial_{\tilde{\bm \sigma}_{kl}} \partial_{\tilde{\bm \sigma}_{mn}})\int d^{2M}{\bm a} ~ e^{-\frac{1}{2}\bm{a}^\intercal{\tilde{\bm \sigma}} \bm{a}}\nonumber\\
	& = \frac{1}{4} \Tr [{\tilde{\bm \sigma}}\partial_{\theta}{\bm \sigma}]^2
	- \frac{\Tr [{\tilde{\bm \sigma}}\partial_{\theta}{\bm \sigma}] \partial_{\theta}{\tilde{\bm \sigma}}_{kl}}{\sqrt{\det{\bm{\sigma}}}}(\partial_{\tilde{\bm \sigma}_{kl}}) \det {\tilde{\bm \sigma}}^{-\frac{1}{2}}
	+ \frac{\partial_{\theta}{\tilde{\bm \sigma}}_{kl}\partial_{\theta}{\tilde{\bm \sigma}}_{mn}}{\sqrt{\det{\bm{\sigma}}}}(\partial_{\tilde{\bm \sigma}_{kl}} \partial_{\tilde{\bm \sigma}_{mn}})\det {\tilde{\bm \sigma}}^{-\frac{1}{2}}\nonumber\\
	& = \frac{1}{4} \Tr [{\tilde{\bm \sigma}}\partial_{\theta}{\bm \sigma}]^2 
	+\frac{1}{2} {\Tr [{\tilde{\bm \sigma}}\partial_{\theta}{\bm \sigma}] \partial_{\theta}{\tilde{\bm \sigma}}_{kl}} [{\tilde{\bm \sigma}}^{-1}]_{kl}
	- \frac{\partial_{\theta}{\tilde{\bm \sigma}}_{kl}\partial_{\theta}{\tilde{\bm \sigma}}_{mn}}{2\sqrt{\det{\bm{\sigma}}}}(\partial_{\tilde{\bm \sigma}_{mn}})([{\tilde{\bm \sigma}}^{-1}]_{kl}\det {\tilde{\bm \sigma}}^{-\frac{1}{2}})\nonumber\\
	& = \frac{1}{4} \Tr [{\tilde{\bm \sigma}}\partial_{\theta}{\bm \sigma}]^2 
	+ \frac{1}{2} \Tr [{\tilde{\bm \sigma}}\partial_{\theta}{\bm \sigma}] \Tr [{\bm \sigma}\partial_{\theta}\tilde{\bm \sigma}]
	+ \partial_{\theta}{\tilde{\bm \sigma}}_{kl}\partial_{\theta}{\tilde{\bm \sigma}}_{mn}\left(\frac{1}{4}[{\tilde{\bm \sigma}}^{-1}]_{kl}[{\tilde{\bm \sigma}}^{-1}]_{mn}
	+ \frac{1}{2}[{\tilde{\bm \sigma}}^{-1}]_{mk}[{\tilde{\bm \sigma}}^{-1}]_{nl} \right) \nonumber\\
	& = \frac{1}{4} \Tr [{\tilde{\bm \sigma}}\partial_{\theta}{\bm \sigma}]^2 
	- \frac{1}{2} \Tr [{\tilde{\bm \sigma}}\partial_{\theta}{\bm \sigma}]^2 
	+ \frac{1}{4} \Tr [{\bm \sigma}\partial_{\theta}{\tilde{\bm \sigma}}]^2
	+ \frac{1}{2} \Tr [(\partial_{\theta}{\tilde{\bm \sigma}}) {\bm \sigma} (\partial_{\theta}{\tilde{\bm \sigma}}) {\bm \sigma}]\nonumber\\
	&=\frac{1}{2} \Tr [ {\bm \sigma}^{-1} (\partial_{\theta}{\bm \sigma}) {\bm \sigma}^{-1}(\partial_{\theta}{\bm \sigma})] 
	= \frac{1}{2} \Tr [ ({\bm \sigma}^{-1} \partial_{\theta}{\bm \sigma})^2]
\end{align}
where we benefited multiple uses of 
$
{\tilde{\bm \sigma}}\partial_{\theta}{\bm \sigma} = -{\bm \sigma}\partial_{\theta}\tilde{\bm \sigma}$, and also 
\begin{align}
\partial_{\tilde{\bm \sigma}_{mn}}[{\tilde{\bm \sigma}}^{-1}]_{kl} &= -[{\tilde{\bm \sigma}}^{-1} (\partial_{\tilde{\bm \sigma}_{mn}}{\tilde{\bm \sigma}}) {\tilde{\bm \sigma}}^{-1}]_{kl} = - [{\tilde{\bm \sigma}}^{-1}]_{kk^{\prime}} [\partial_{\tilde{\bm \sigma}_{mn}}{\tilde{\bm \sigma}}]_{k^{\prime}l^{\prime}}[{\tilde{\bm \sigma}}^{-1}]_{l^{\prime}l} = - [{\tilde{\bm \sigma}}^{-1}]_{kk^{\prime}} \delta_{mk^{\prime}}\delta_{nl^{\prime}}[{\tilde{\bm \sigma}}^{-1}]_{l^{\prime}l} \nonumber\\
&= - [{\tilde{\bm \sigma}}^{-1}]_{km}[{\tilde{\bm \sigma}}^{-1}]_{nl}.
\end{align}
%--------------
%--------------
\section{The maximum likelihood estimator}\label{sec:MLE}
In this Appendix, through an example we provide details on how one can actually process experimental outcomes to assign an estimate to the unknown parameter---namely temperature. In local parameter estimation the maximum likelihood estimator (MLE) is known to saturate the Cram\'er-Rao bound~\cite{paris2009quantum}. In global estimation strategies, however, the MLE saturates the bound in the asymptotic limit of large number of repetitions of the measurement. Nonetheless, it is well-know that the MLE can saturate the bound when the outcome probability distribution is described by the ``Gaussian family''---see e.g., section 3.1.2.4 of ~\cite{kolodynski2014precision} and references therein. This is indeed the case for our Gaussian measurements.

The Likelihood function is defined as the probability of observing the outcome ``${\bm a}$'' if we performed some measurement ${\bm \sigma}^M_s$ on the state ${\bm \sigma}$. By using Eq.~\eqref{eq:prob_comp} we have
\begin{align}
    L(T|{\bm a};{\bm \sigma}(T),{\bm d};{\bm \sigma}^M_s) \coloneqq p({\bm a}|{\bm \sigma}(T),{\bm d};{\bm \sigma}_s^M) =   \frac{e^{-\frac{1}{2}({\bm a}-{\bm d})^\intercal {({\bm \sigma}(T)+{\bm \sigma}^{M}_{s})}^{-1}({\bm a}-{\bm d})}}{2\pi\sqrt{{\rm det} ({\bm \sigma}(T)+{\bm \sigma}^{M}_{s})}}.
\end{align}
The MLE, assigns the estimate ${\tilde T}$ to the unknown parameter that maximises the likelihood or the log likelihood. That is
\begin{align}
    & \partial_{T} \log L(T|{\bm a};{\bm \sigma}(T),{\bm d};{\bm \sigma}^M_s)|_{\tilde T} =0 \nonumber\\
    & \Rightarrow 
    -\frac{1}{2}\partial_{T} \log({\rm det} ({\bm \sigma}(T)+{\bm \sigma}^{M}_{s}))|_{\tilde T} -\frac{1}{2}~({\bm a}-{\bm d})^\intercal~\partial_{T}|({\bm \sigma}(T)+{\bm \sigma}^{M}_{s})^{-1}~({\bm a}-{\bm d}) |_{\tilde T} = 0
\end{align}
By using the Jacobie's formula for derivative of the determinant, and the formula for derivative of inverse matrix, this condition reduces to
\begin{align}
    ({\bm a}-{\bm d})^\intercal~({\bm \sigma}(T)+{\bm \sigma}^{M}_{s})^{-1} [\partial_T{\bm \sigma}(T)]({\bm \sigma}(T)+{\bm \sigma}^{M}_{s})^{-1}~({\bm a}-{\bm d})|_{\tilde T} = {\rm Tr}[({\bm \sigma}(T)+{\bm \sigma}^{M}_{s})^{-1} {\bm \sigma}_s^{\prime}]_{\tilde T}.
\end{align}
Now let's consider a single mode scenario in which the measurement is Homodyne detection with ${\bm \sigma}_s^M = 1/2\lim_{r\to 0}{\rm diag}[r,~1/r]$. In this case, we are projecting the state into the first quadrature. The above condition reduces to 
\begin{align}
    \frac{\partial_{T}[{\bm \sigma}(T)]_{1,1} }{[{\bm \sigma}(T)]_{1,1}^2}\left({\bm a}_1-{\bm d}_1\right)^2 \Big|_{\tilde T}= \frac{\partial_{T}[{\bm \sigma}(T)]_{1,1}}{[{\bm \sigma}(T)]_{1,1}}\Big|_{\tilde T}
    \Rightarrow \left({\bm a}_1-{\bm d}_1\right)^2 = [{\bm \sigma}({\tilde T})]_{1,1},
\end{align}
which is quite expected.

In order to convert this to a temperature estimate, we can use the parameter dependence of ${\bm \sigma}(T)$.
As a relevant example, when the state is at thermal equilibrium, we have ${\bm \sigma}(T) = 1/2 \coth(\omega/2T) {\rm diag}[1/\omega^2, 1]$. As such, upon observing some outcome ``${\bm a}_1$'', we assign the following estimate to the temperature
\begin{align}
    {\tilde T} = \frac{\omega}{2 {\rm acoth}\Big(2\omega^2 \left({\bm a}_1-{\bm d}_1\right)^2\Big)}.
\end{align}
Similarly, if we repeat the measurement $k$ times, we have
\begin{align}
    {\tilde T} = \frac{\omega}{2 {\rm acoth}\Big(2\omega^2 \average{\left({\bm a}_1-{\bm d}_1\right)^2}\Big)},
\end{align}
with the angled brackets in $\average{\left({\bm a}_1-{\bm d}_1\right)^2}$ being the average over the $k$ measurement repetitions.
In Fig.~\ref{fig:MLE_HD} we demonstrate a simulation of thermometry of a single mode thermal state by using Homodyne detection. One can observe how the estimate gets closer and closer to the true value of the temperature as we increase $k$---a direct result of the central limit theorem. Moreover, one can see that the estimation error follows the trend of the Cram\'er-Rao bound.
\begin{figure}
    \centering
    \includegraphics[width=.5\linewidth]{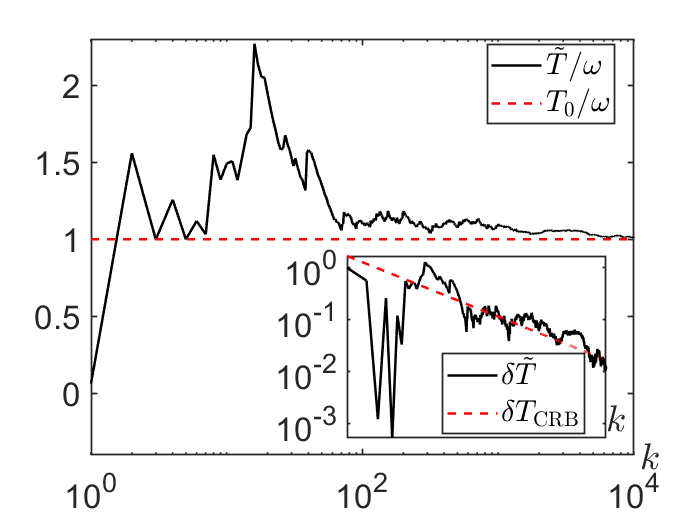}
    \caption{Loglog plot illustrating thermometry of a single Harmonic oscillator with Homodyne detection. The solid black line shows a (single) trajectory of the temperature estimate---calculated by simulating the measurement outcomes $a_1$ according to the true temperature $T_0$, and processing the outcomes by using the maximum likelihood estimator. One can clearly see that the estimate converges to the true temperature $T_0$ as we increase the measurement repetitions $k$. In the inset, we depict the relative error $\delta {\tilde T}\coloneqq ({\tilde T}-T_0)/T_0$ and benchmark it against the Cram\'er-Rao bound $\delta T_{\rm CRB} = 1/\sqrt{k {\cal F}({\bm \sigma}(T);{\bm \sigma}^M_{\rm HD})}$. Here we set $\omega=1$. }
    \label{fig:MLE_HD}
\end{figure}
%--------------
%--------------
%------------------------
\section{Imperfect Homodyne detection}\label{App:HD}
%------------------------
\begin{figure}[H]
    \centering
    \includegraphics[width=.5\linewidth]{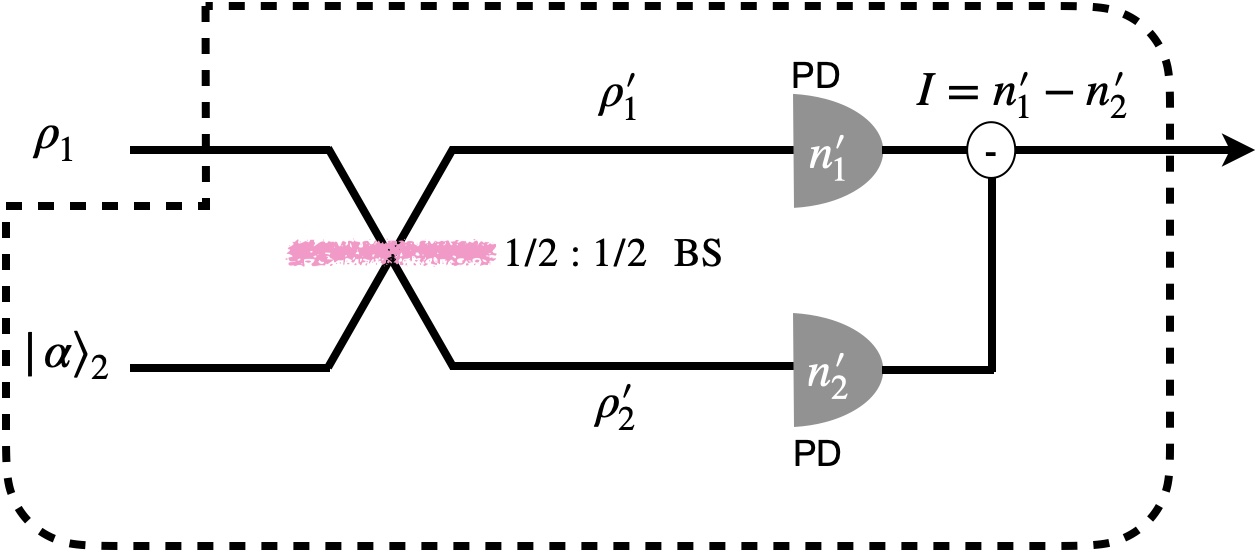}
    \caption{Schematic of Homodyne detection; that is everything within the dashed box. The signal $\rho_1$ enters from the upper arm, and interferes with a local oscillator (a coherent state) in a $1/2:1/2$ beam splitter. The final output is the difference of the photocurrent from the two photodetectos, i.e., $I = n_1^{\prime}-n_2^{\prime}$. If the local oscillator is highly displaced in either of the quadratures, say $\average{q_2}\gg 1$ while keeping the other quadrature zero i.e., $\average{p_2}=0$, then the outcome is proportional to the $q_1$ quadrature of the signal, i.e., $I\propto q_1$. Similarly one can measure the other quadrature. However, if the local oscillator has a limited energy, such a perfect quadrature measurement might be affected. For thermometry task at thermal equilibrium this imperfection adds a bias to the temperature estimate.}
    \label{fig:HD_setup}
\end{figure}
%------------------------
%------------------------
Following~\cite{PhysRevA.42.474}, we note that measuring the quadratures may be implemented by Homodyne detection, that is by interfering the signal and a local oscillator through a balanced beam splitter. The output modes will be photodetected and the currents will be subtracted (see Fig.~\ref{fig:HD_setup}). Perfect quadrature measurement requires a signal that has infinite energy. In many practical situations, the energy of the local oscillator is orders of magnitude higher than the signal which leads to very precise quadrature measurement---see e.g., \cite{lodewyck:tel-00130680}. Nonetheless, in theory it is not be possible to have unlimited energy; and one can study its impact on the measurement precision. Below we examine the correction due to finite energy.

Firstly, note that the beam splitter is a Gaussian operation which can be represented by the following symplectic transformation
\begin{align}
    {\bm S}_{\rm BS} =\frac{1}{\sqrt{2}} \left[\begin{array}{cc}
    \cos(\theta){\bm I}_2 & \sin(\theta){\bm I}_2\\
    \sin(\theta){\bm I}_2 & -\cos(\theta){\bm I}_2
    \end{array}
    \right],
\end{align}
where ${\bm I}_2$ represents the $2\times 2$ identity matrix. The beam splitter transforms the quadratures as follows
\begin{align}
    R \mapsto R^{\prime} = {\bm S}_{\rm BS} R = \frac{1}{\sqrt{2}}\left[
    \begin{array}{c}
    \cos(\theta)q_1+\sin(\theta)q_2\\
    \cos(\theta)p_1+\sin(\theta)p_2\\
    \sin(\theta)q_1-\cos(\theta)q_2\\
    \sin(\theta)p_1-\cos(\theta)p_2
    \end{array}
    \right],
\end{align}
with the first two quadrature representing the signal, and the second two representing the local oscillator. 
For a balanced beam splitter we have $\theta = \pi/4$ which leads to
\begin{align}
    {\bm S}_{\rm BS} =\frac{1}{\sqrt{2}} \left[\begin{array}{cc}
    {\bm I}_2 & {\bm I}_2\\
    {\bm I}_2 & -{\bm I}_2
    \end{array}
    \right],
\end{align}
which transforms the quadratures as follows
\begin{align}\label{eq:BS_quadrature_balanced}
    R \mapsto R^{\prime} = {\bm S}_{\rm BS} R = \frac{1}{\sqrt{2}}\left[
    \begin{array}{c}
    q_1+q_2\\
    p_1+p_2\\
    q_1-q_2\\
    p_1-p_2
    \end{array}
    \right],
\end{align}
The operator that represents the difference of the photocurrents from the two arms is given by 
\begin{align} 
I \coloneqq n_1^{\prime} - n_2^{\prime} =  q_1q_2 + p_1p_2, 
\end{align}
where we defined the number operators $n_k = (q_k-ip_k)(q_k+ip_k)/2~ k\in\{1,2\}$ and similarly $n_k^{\prime} = (q_k^{\prime}-ip_k^{\prime})(q_k^{\prime}+ip_k^{\prime})/2~ k\in\{1,2\}$.
Note that this operator is generally not Gaussian.
Nonetheless, we note that by choosing the local oscillator such that $\average{p_2}=0$ we have
\begin{align}
    \average{I} = \average{q_1}
    \average{q_2},\label{eq:HD_av}
\end{align}
which up to a constant reproduces the first moment statistics of the position quadrature.
However, as we saw in Appendix~\ref{sec:MLE}, the estimator for thermometry is built upon the second moment, and thus ${\rm Var}(q_1)$ plays the most crucial role in quality of Homodyne detection for thermometry at thermal equilibrium.
For the second moment, we have
\begin{align}
{\rm Var}(I) & = \average{I^2} - \average{I}^2 =  \average{(q_1q_2 + p_1p_2)^2} - \average{q_1}^2\average{q_2}^2 \nonumber\\
& = \average{q_1^2q_2^2 +p_1^2p_2^2 + q_1p_1q_2p_2 + p_1q_1p_2q_2 } - \average{q_1}^2\average{q_2}^2 \nonumber\\
&= \average{q_1^2}\average{q_2^2} + \average{p_1^2}\average{p_2^2} -\average{q_1}^2\average{q_2}^2 -\frac{1}{2},
\end{align}
where we used $[q_2,p_2]=i$ and we assumed the local oscillator has a diagonal covariance matrix, i.e., $\average{\{q_2,p_2\}} = 0$.

In the last equation, we have some extra terms in ${\rm Var}(I)$ that contain information about the momentum quadrature. To get rid of them, we prepare the local oscillator in a coherent state with the covariance matrix ${\bm \sigma}_{\rm LO} = {\bm I}/2$ and with the displacement ${\bm d}_{\rm LO}=(\average{q_2}, 0)^{\intercal}$. We have
\begin{align}
   {\rm Var}(I) & = {\rm Var}(q_1)\average{q_2}^2 + \frac{1}{2}\left(\average{q_1}^2 + \average{p_1^2}-1
   \right),\label{eq:HD_var}
\end{align}
where we used that $\average{p_2^2}=1/2$. Equations \eqref{eq:HD_av} and \eqref{eq:HD_var} suggest that $I/\average{q_2}$ has the same mean value as the position quadrature $q_1$, while its variance is modified by an additional $\frac{1}{2\average{q_2}^2}\left(\average{q_1}^2 + \average{p_1^2}-1
   \right)\eqqcolon \Delta$.
In the limit of $\average{q_2}^2 \to \infty$ we can reproduce the quadrature statistics perfectly by using $q_1 = I/\average{q_2}$. However, this requires the local oscillator to have infinite energy. Nonetheless, for any situation that ${\rm Var}(q_1)\gg \Delta$, we can still have a good approximation---in such a case our estimate for temperature will have a bias that is proportional to $\Delta/{\rm Var}(q_1)$. While for local metrology schemes one can correct this bias (since we can exactly calculate it), for global thermometry protocols this is not possible.

\vspace{.5cm}
\section{Additivity of the Gaussian-measured Fisher information in a special case}

\begin{proposition}
Let ${\bm \sigma}(T)$ be a single-mode covariance matrix, and assume that ${\bm \sigma} = \nu {\bm I}_2$ and $\partial_T{\bm \sigma}=\mu {\bm I}_2$ are both proportional to the identity. Then for all $n$ it holds that
\begin{equation}
    \max_{{\bm \sigma}^M_s} {\cal F}^{\rm C}\left({\bm \sigma}^{\oplus n}; {\bm \sigma}^M_s\right) = n \max_{{\bm \gamma}^M_s} {\cal F}^{\rm C}\left( {\bm \sigma}; {\bm \gamma}^M_s\right) = n \mu^2 \max\left\{\frac{1}{\nu^2},\, \frac{2}{(1+\nu)^2} \right\} ,
\end{equation}
where the optimization on the left-hand side is over all (multimode) covariance matrices of the measurement ${\bm \sigma}^M_s$.
\end{proposition}

\begin{proof}
Let ${\bm \gamma}^M_s$ be an arbitrary covariance matrix of an $n$-mode Gaussian pure state. Then the spectrum of ${\bm \gamma}^M_s$ is well known to be of the form $\kappa_1,\ldots, \kappa_n,\kappa_n^{-1},\ldots, \kappa_1^{-1}$, where $\kappa_1,\ldots,\kappa_n\geqslant 1$. Thus,
\begin{align*}
    {\cal F}^{\rm C}\left({\bm \sigma};{\bm \sigma}_s^{M}\right) &= \frac{1}{2}\Tr \left[\left(\left({\bm \sigma}^{\oplus n} + {\bm \gamma}_s^{M}\right)^{-1} (\partial_{T}{\bm \sigma})^{\oplus n}\right)^2 \right] \\
    &= \frac{\mu^2}{2} \Tr \left[\left(\nu {\bm I}_{2n} + {\bm \gamma}_s^{M}\right)^{-2} \right] \\
    &= \frac{\mu^2}{2} \sum_{j=1}^n \left( \left(\nu+\kappa_j\right)^{-2} + \left(\nu+\kappa^{-1}\right)^{-2} \right) .
\end{align*}
Hence, maximizing over ${\bm \gamma}^M_s$, that is to say, maximizing over all choices of $\kappa_j$ subjected to the above constraints, we obtain that
\begin{align*}
    \max_{{\bm \sigma}^M_s} {\cal F}^{\rm C}\left({\bm \sigma}^{\oplus n}; {\bm \sigma}^M_s\right) &= \max_{\kappa_1, \ldots, \kappa_n\geqslant 1} \frac{\mu^2}{2} \sum_{j=1}^n \left( \left(\nu+\kappa_j\right)^{-2} + \left(\nu+\kappa^{-1}\right)^{-2} \right) \\
    &= \frac{\mu^2}{2} \sum_{j=1}^n \max_{\kappa_j\geqslant 1} \left( \left(\nu+\kappa_j\right)^{-2} + \left(\nu+\kappa^{-1}\right)^{-2} \right) \\
    &= \frac{n \mu^2}{2} \max\left\{\frac{1}{\nu^2},\, \frac{2}{(1+\nu)^2} \right\} \\
    &= n \max_{{\bm \gamma}^M_s} {\cal F}^{\rm C}\left( {\bm \sigma}; {\bm \gamma}^M_s\right) .
\end{align*}
This establishes the additivity of the Gaussian-measured Fisher information in this special case.
\end{proof}
%-------------
%-------------
\section{Numerical optimisation, and the algorithm}
We aim to identify the measurement CM ${\bm \sigma}_{s}^M$ that maximises \eqref{eq:CFI_Gaussian_main}.
We are dealing with a nonlinear optimisation problem subject to the positive semi-definiteness criterion ${\bm \sigma}_{s}^M + i{\bm \Omega} \geqslant 0$. We already proved that we should have ${\bm \sigma}_{s}^M={\bm S}^M({\bm S}^{M})^{\intercal}$, with some symplectic transformation satisfying ${\bm S}^M{\bm \Omega}({\bm S}^{M})^{\intercal}={\bm \Omega}$. One can find the optimal ${\bm S}^M$ in MATLAB by using the code at the end of this section. Notice that this algorithm works for arbitrary parameter estimation problem and is not limited to thermometry. One just needs to replace ``${\rm dsigma}$'' with the derivative of the covariance matrix w.r.t. the parameter to be estimated.

\subsection{Examples of the simulation}
As we mentioned in the main text, our numerical simulations support the conjecture~\eqref{eq:equivalence_conj}. Here, we present two such simulations: (i) a scenario with two modes with different frequencies for different temperatures, and (ii) a scenario with several oscillators, each with a different frequency but for a finxed temperature. As one can see in Fig.~\ref{fig:CFI_local_vs_Global_vs_T}, in either of these simulations global Gaussian measurements have no advantage over local ones.
\begin{figure}
    \centering
    \includegraphics[width=.49\linewidth]{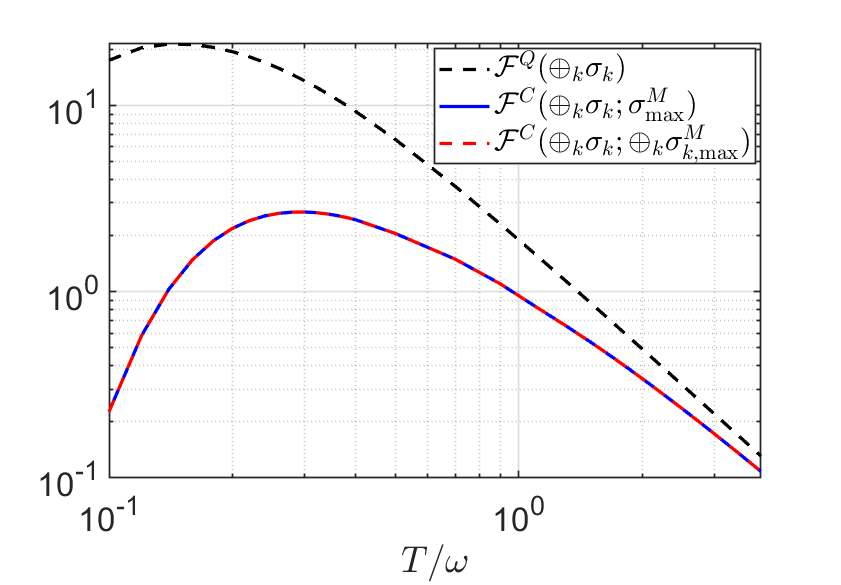}
    \includegraphics[width=.49\linewidth]{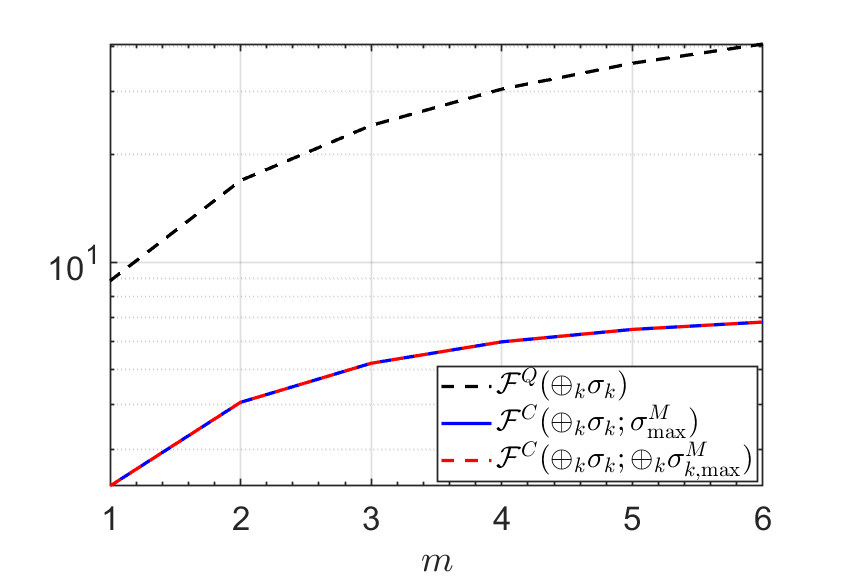}
    \caption{Testing the conjecture \eqref{eq:equivalence_conj}. Left---Two harmonic oscillators at frequencies $\omega_1=0.5\omega$ and $\omega_2=\omega$. Right---A scenario with $m$ harmonic oscillators, with frequencies $\omega_m=(0.5+0.1 m)\omega$, and for the fixed temperature $T=0.3\omega$. Both panels show no advantage in using the global Gaussian operations. Here we have set $\omega=1$.}
    \label{fig:CFI_local_vs_Global_vs_T}
\end{figure}
%%%%%%%%%%%%%%%%%%%%%%%%%%%%%%%%%%%%%%%%%%%%%
%\pagebreak
\begin{lstlisting}
function [CFI,sigma_s] = optim_Gaussian(sigma_modes,dsigma_modes,num_modes) 
%-------------------
%   sigma_modes is the Cov. Mat. of the system, which has 
%   `num_modes' modes 
%   dsigma_modes is the derivative of sigma_modes w.r.t the parameter
%-------------------
%The Symplectic form
Omega = [0 1;-1 0];Omega = kron(eye(num_modes),Omega);
%Objective: (minus) the Fisher info. 
%Here, the measurement CM can be written as sigma^M=SS^T
fisher = @(S) -1/2*trace(  (dsigma_modes...
    /(sigma_modes + S*transpose(S)))^2 );
%The constraint on S, i.e., it should be symplectic 
nonlequ = @(S) S*Omega*transpose(S)-Omega;
%the [] means we don't have inequalities as constraints. 
%nonlequ(S) is an equality constraint
nonlcon = @(S) deal([],nonlequ(S));
%Initial guess for the optimisation variable; 
%Probably a better guess is the locally optimal measurement
S0 = eye(2*num_modes);
%fmincon options; change to suite your problem, if needed.
%Specifically, 'ConstraintTolerance' should be small enough to guarantee
%physicality of the answer.
options = optimoptions('fmincon', 'StepTolerance', 1e-3,...
    'FiniteDifferenceStepSize', 1e-3,'MaxFunctionEvaluations',10^5,...
    'MaxIterations',10^3,'ConstraintTolerance',1e-15,...
    'Algorithm','interior-point');
problem = createOptimProblem('fmincon','objective',...
    fisher,'x0',S0,'nonlcon',nonlcon,'options',options);
ms = MultiStart('FunctionTolerance',2e-4,'UseParallel',true);
gs = GlobalSearch(ms,'NumTrialPoints',1e4);
[S,fval] = run(gs,problem);
%The optimal measurement's Cov. Mat.
sigma_s = S*transpose(S);
%The Fisher information corresponding to the optimal measurement
CFI=-fval;
end
\end{lstlisting}
%\newpage

%--------------------------------------------------------------------------------------
%--------------------------------------------------------------------------------------
%	THE END
%--------------------------------------------------------------------------------------
%--------------------------------------------------------------------------------------
%\newpage
%\bibliographystyle{apsrev4-1}
\bibliography{Refs}
\end{document}